\documentclass[11pt]{article}
\usepackage[left=1in, right=1in, top=1in, bottom=1in, margin=1in]{geometry}
\usepackage{latexsym,amssymb,amsfonts,amsmath,mathrsfs,hyperref,amsthm,thmtools,thm-restate}

\usepackage{tikz}
\usepackage{pifont}%
\usepackage{enumerate}
\hypersetup{colorlinks,citecolor=NavyBlue,filecolor=black,linkcolor=black,urlcolor=black}
\usepackage[dvipsnames]{xcolor}
\usepackage{comment} 
\usepackage[capitalize]{cleveref}

\newcommand{\vertiii}[1]{{\left\vert\kern-0.25ex\left\vert\kern-0.25ex\left\vert #1 
		\right\vert\kern-0.25ex\right\vert\kern-0.25ex\right\vert}}

  \DeclareMathOperator{\var}{Var}
\DeclareMathOperator{\maxinf}{MaxInf}
  \renewcommand{\Pr}{\mbox{\rm Pr}}
  
  \newcommand{\E}{\mathbb{E}}

  \DeclareMathOperator{\cb}{cb}
    
    \DeclareMathOperator{\op}{op}

  \newcommand{\norm}[1]{\|#1\|}

  \newcommand{\R}{\mathbb{R}} 
  \newcommand{\N}{\mathbb{N}} 

  \newcommand{\Id}{\ensuremath{\mathop{\rm Id}\nolimits}}

 \newcommand\Inf{\mathrm{Inf}}

  \newcommand{\eps}{\varepsilon}

  \DeclareMathOperator{\sign}{sign}

  \DeclareMathOperator{\Tr}{Tr}
  
  \DeclareMathOperator{\Diag}{Diag}

  \newcommand{\ket}[1]{|#1\rangle}

  \newcommand{\ind}[1]{\mathbf{#1}}
   
  \newcommand{\beq}{\begin{equation}}
  \newcommand{\eeq}{\end{equation}}
  \newcommand{\beqn}{\begin{equation*}}
  \newcommand{\eeqn}{\end{equation*}}
  \newcommand{\beqr}{\begin{eqnarray}}
  \newcommand{\eeqr}{\end{eqnarray}}
  \newcommand{\beqrn}{\begin{eqnarray*}}
  \newcommand{\eeqrn}{\end{eqnarray*}}
  \newcommand{\bmline}{\begin{multline}}
  \newcommand{\emline}{\end{multline}}
  \newcommand{\bmlinen}{\begin{multline*}}
  \newcommand{\emlinen}{\end{multline*}}

  \theoremstyle{plain}
  \newtheorem{theorem}{Theorem}[section]
  \newtheorem{lemma}[theorem]{Lemma}
  \newtheorem{fact}[theorem]{Fact}
  \newtheorem{proposition}[theorem]{Proposition}
  \newtheorem{claim}[theorem]{Claim}
  \newtheorem{corollary}[theorem]{Corollary}
  
  \theoremstyle{definition}

  \newtheorem{conjecture}[theorem]{Conjecture}
  
  \theoremstyle{remark}
  \newtheorem{remark}[theorem]{Remark}
  
  \renewenvironment{proof}[1][]{
    	\begin{trivlist}
     	\item[\hspace{\labelsep}{\em\noindent Proof#1:\/}]}
     	{{\hfill$\Box$}
    	\end{trivlist}
  }
  \newenvironment{claimproof}[1][]{
  		\begin{trivlist}
   		\item[\hspace{\labelsep}{\sc\noindent Proof #1:\/}]}
   		{{\hfill$\blacklozenge$}
  		\end{trivlist}
  }
\title{
Optimal inequalities for completely bounded polynomials\\ and the limitations of quantum query algorithms
} 

\date{\vspace{-1cm}}

\author{Francisco Escudero Guti\'errez\\ \texttt{Quriosity, Inria, France} \\ \texttt{fescuder@zimbra.inria.fr}
\and
Miquel Saucedo\\ \texttt{Centre de Recerca Matemàtica, Spain} \\ \texttt{msaucedo@crm.cat}
\and
Carlos Palazuelos\\ \texttt{Universidad Complutense de Madrid, Spain} \\ \texttt{cpalazue@ucm.es}}
\begin{document}

\maketitle

\begin{abstract}
    We consider the problem of establishing limitations on the power of quantum query algorithms. 
    We do so via a refinement of the polynomial method due to Arunachalam, Briët and Palazuelos (SICOMP, 2019), who showed that quantum query algorithms are completely bounded polynomials. Motivated by that, we prove several optimal functional inequalities involving different notions of completely bounded polynomials. These inequalities lead to limiting theorems for the power of quantum query algorithms that improve on several prior works. In particular, we prove the following.

    \textbf{1. An optimal root-influence bound  for block-multilinear polynomials.} Prior work showed that block-multilinear polynomials $p$ of degree $t$ satisfy a root-influence bound, $\norm{p}_{\cb}\geq \sum_i \sqrt{\Inf_i[p]}/t^2$, which is stronger than the bound appearing in the Aaronson--Ambainis conjecture. We improve the constant in that inequality, proving its optimal version: $\norm{p}_{\cb}\geq \sum_i \sqrt{\Inf_i[p]}/t$. Since the amplitudes of quantum algorithms that query disjoint blocks of inputs—such as $t$-fold forrelation— are block-multilinear polynomials with  $\norm{p}_{\cb}\leq 1,$ our inequality shows that they satisfy $t\geq \sum_i\sqrt{\Inf_i[p]}$. We prove that this stronger inequality yields both a more efficient classical simulation for these algorithms than prior results based on the Aaronson–Ambainis argument, and a qualitative improvement: all classical queries can be made in a non-adaptive manner.

    \textbf{2. Optimal Fourier growth of the highest level of quantum query algorithms.} We show that for every polynomial $p$ defined on $\{-1,1\}^n$ of degree $2t$, the Fourier Growth at the level $2t,$ namely $\norm{\widehat p_{2t}}_{\ell_1}$, satisfies $\norm{\widehat p_{2t}}_{\ell_1}\leq (en/(2t-1))^{\frac{2t-1}{2}}\norm{p}_{\cb}$. This is optimal up to the factor $e$, as witnessed by $2t$-fold forrelation. As the acceptance probabilities of quantum query algorithms that make $t$ queries (to the whole input) satisfy $\norm{p}_{\cb}\leq 1$, this yields a Fourier growth bound for these algorithms, partially resolving a question by Girish~(STOC, 2026).

\end{abstract}

\section{Introduction}
Since the early days of quantum computing, query complexity has served as a rigorous model to compare the power of classical and quantum computers \cite{Ambainis:2018}. In this model, the unit of computational complexity is a query, which is a call to an oracle that reveals a unit of information about the input. More precisely, in the case of Boolean functions $f:D\subseteq \{-1,1\}^n\to\{-1,1\}$, query algorithms have access to a full description of $f$ and are requested to compute $f$ on an unknown input $x\in D$. They can access this input via oracle calls. In the classical setting, each oracle call reveals an entry $i$ of the input, while in the quantum case an oracle is a call to (the controlled version of) the unitary that maps $\ket{i}$ to $x_i\ket{i}$. Many celebrated quantum algorithms show an advantage in terms of query complexity, for example in unstructured search \cite{Grover:1996}, period finding \cite{Shor:1997}, Simon’s problem \cite{Simon}, NAND-tree evaluation \cite{farhi2007quantum} and element distinctness \cite{ambainis2007quantum}. 

Studying the power of query algorithms to compute Boolean functions $f:D\subseteq \{-1,1\}^n\to \{-1,1\}$ has been closely linked to Fourier analysis. Since the seminal work of Beals, Buhrman, Cleve, Mosca and de Wolf, it has been known that the acceptance probability (the probability of outputting 1) of a quantum algorithm that makes $t$ queries is a polynomial $p:\{-1,1\}^n\to\R$ of degree at most $2t$ in its Fourier expansion \cite{beals2001quantum}. Furthermore, such polynomials are also bounded in the infinity norm, i.e. they satisfy $\norm{p}_\infty=\sup_{x\in\{-1,1\}^n}|p(x)|\leq 1.$ This structural fact 
yields the quantum polynomial method to lower bound quantum query complexity. Indeed, given a function $f:D\subseteq \{-1,1\}^n\to\{-1,1\}$, assume that one can show that every polynomial $p:\{-1,1\}^n\to [-1,1]$ that approximates $f$ up to error $1/3$ needs to have degree at least $t.$ Then, by the remark of Beals et al.\, any quantum algorithm that guesses the value of $f$ with probability $\geq 1/3$ must make at least $t/2$ queries. Throughout the years, this strategy has led to several breakthroughs~\cite{aaronson2004quantum,BunKhotariThaler:2020,mande2020improved,bun2023approximate}.

\vspace{0.3cm}
\noindent\textbf{Aaronson--Ambainis conjecture.} In addition, the polynomial method has been used to establish other kinds of limitations of quantum query algorithms, such as in the Aaronson--Ambainis conjecture \cite{Aaronsons:2014}, which aims to delineate the landscape of exponential quantum speedups for Boolean functions. It is widely believed that to achieve superpolynomial quantum speedups for computing Boolean functions, the domain $D$ of the function must be very structured. More precisely, the community expects that, to obtain such speedups, the domain must satisfy $|D|=o(2^n),$ in other words,  the input $x$ needs to be promised not to be an arbitrary element in $\{-1,1\}^n$. This fact can be exploited by quantum computers better than by classical ones (see, for instance, the Deutsch-Jozsa problem \cite{deutsch1992rapid}). Such a belief is formalized in the following folklore conjecture dating from the 90s. 

\begin{conjecture}[Simulation conjecture]\label{con:needforstructure}
	The acceptance probability of $t$-query quantum algorithms can be simulated with error at most $\eps$ on at least a ($1-\delta$)-fraction of the inputs using poly($t,1/\eps,1/\delta$) classical queries.
\end{conjecture}

\cref{con:needforstructure} implies that exponential quantum speedups need structure. Indeed, a quantum query algorithm is naturally defined on the whole Boolean cube $\{-1,1\}^n$, and if \cref{con:AAconjecture} is true, then it can be efficiently classically simulated on a large part of $\{-1,1\}^n$. As a result, for such an algorithm to achieve an exponential speedup when computing a partial function $f:D\to \{-1,1\}$, its domain $D$ must have a low overlap with the subset where the classical simulation succeeds, and thus $D$ must be small. 

Motivated by the lack of exponential quantum speedups without structure and by a result of Fourier analysis \cite{dinur2006fourier}, Aaronson and Ambainis posed a functional analytic conjecture that implies \cref{con:needforstructure} \cite{Aaronsons:2014}. To state the conjecture, we must recall the definition of the influences and variance of a polynomial $p:\{-1,1\}^n\to \R.$ Given such a $p,$ its variance is given by $\var [p]=\E_{x}[(p(x)-\E_x[p(x)])^2],$ where this expectation, and all to come, are taken with respect to the uniform measure on the Boolean cube. The influence of the $i$-th variable is $\Inf_i[p]=\E_{x}[(p(x)-p(x^{\oplus i}))^2/4],$ where $x^{\oplus i}$ is the element of the Boolean cube obtained by flipping the $i$-th entry of~$x.$ Finally, the maximum influence is $\maxinf[p]=\max_{i\in [n]}\Inf_i[p].$ Both quantities can be expressed in terms of the Fourier coefficients appearing in the Fourier expansion of $p=\sum_{S\subseteq [n]}\widehat p(S)\chi_S,$ where $\chi_S(x)=\prod_{i\in S}x_i.$ Indeed, one can check that $\Inf_i[p]=\sum_{i\ni S}|\widehat p(S)|^2,$ and $\var[p]=\sum_{\emptyset \neq S\subseteq [n]}|\widehat p(S)|^2.$ Now, we are ready to state the Aaronson--Ambainis conjecture.

\begin{conjecture}[Aaronson--Ambainis]\label{con:AAconjecture}
	Let $p:\{-1,1\}^n\to \mathbb{R}$ be a polynomial of degree at most $d$ with $\norm{p}_\infty\leq 1$. Then, $p$ has a variable with influence at least poly($\var[p],1/d$).
\end{conjecture}

The argument of \cite[Theorem 22]{Aaronsons:2014} to show that \cref{con:AAconjecture} implies \cref{con:needforstructure} works as follows. Let $p$ be the bounded polynomial of degree at most $2t$ that represents the acceptance probability of a $t$-query quantum algorithm, and suppose that we want to approximate $p(y)$ for an unknown $y\in \{-1,1\}^n$. First, query an influential variable $i$ of $p$. Then, the restricted polynomial $p|_{x(i)=y(i)}$ is also a bounded polynomial of degree at most $2t$, so we can query again an influential variable. Assuming that \cref{con:AAconjecture} is true, the influences of these variables are large. As a consequence, after a \emph{small} number of queries the remaining polynomial must have a low variance, so its expectation is close to $p(y)$ with high probability. 

Despite several efforts by experts in the fields of quantum query complexity and Fourier analysis, the Aaronson--Ambainis conjecture remains wide open and it is only known to hold in particular cases \cite{montanaro2012some,o2015polynomial,defant2019fourier,lovett2022,Bansal:2022,gutierrez2023influences,bhattacharya2025random}. We refer to \cite{gutierrez2023influences} for a more detailed review of prior work on the Aaronson--Ambainis conjecture.

\vspace{0.3cm}
\noindent \textbf{Fourier growth.} Another Fourier analytic approach to show limits to the power of query algorithms via the polynomial method is the study of the Fourier growth. Given a polynomial $p:\{-1,1\}^n\to \R$, its Fourier growth at level $d$ is the $\ell_1$-norm of its Fourier coefficients of degree $d$, which we denote by $\norm{\widehat p_{d}}_{\ell_1}.$ Analyzing this quantity allows one to establish separations between different query models by $i)$ showing that for some function computable with few queries in one of the models the corresponding Fourier growth is large, and $ii)$ showing that the functions computed with the other model have low Fourier growth, unless many queries are made. Notably, this meta-strategy was introduced by Raz and Tal to show an oracle separation between BQP and PH \cite{raz2022oracle}. Subsequently, Fourier growth was used to show that $t$-fold forrelation achieves the best possible quantum-classical separation \cite{bansal2021k,girish2021fourier}. Furthermore, Fourier growth has been used to separate different models of quantum query algorithms, such as those with bounded rounds of adaptivity, or noisy models, and to show separations in communication complexity \cite{tal2020towards,girish2023fourier,girish2024power,girish2026fourier}. In addition, Fourier growth has also seen applications in the study of pseudorandomness~\cite{agrawal2020coin,chattopadhyay2018pseudorandom,chattopadhyay2019pseudorandom,chattopadhyay2020fractional}. For a comprehensive review on Fourier growth and query complexity, we refer to the PhD thesis of Wu~\cite{Wu:EECS-2025-39}.

\vspace{0.3cm}
\noindent\textbf{The completely bounded polynomial method.} The polynomial method was born as a framework to prove quantum query lower bounds~\cite{beals2001quantum}. Subsequently, a natural question arose: \emph{Can approximation by bounded polynomials yield quantum query upper bounds?} 
This was ruled out by Ambainis~\cite{Ambainis:2006}. Nonetheless, Aaronson, Ambainis, Iraids, Kokainis, and Smotrovs suggested that by refining the polynomial method, it could yield upper bounds to quantum query complexity \cite{Aaronson2015PolynomialsQQ}. This was confirmed by Arunachalam, Briët and Palazuelos, who precisely characterized the quantum query complexity of a Boolean function $f$ in terms of the \emph{completely bounded norm} of certain multilinear forms \cite{QQA=CBF}. 

Informally speaking, the completely bounded norm of a polynomial is defined as the infinity norm, but evaluating the polynomial in bounded matrices instead of bounded scalars.  
More precisely, the completely bounded norm of a polynomial
$$p:\R^n\to \R:\ x\to\sum_{s\in [t]_0}\sum_{\substack{\ind i\in [n]^s\\\ i_1\leq i_2\leq \dots \leq i_s}} p_{\ind i}x(i_1)\dots x(i_s)$$
is defined as
\begin{equation}\label{eq:cbnorm}
    \norm{p}_{\cb}:=\sup_{m\in \N}\sup_{\substack{X(i)\in\R^{m\times m}\\ \norm{X(i)}_{\op}\leq 1}} \norm{ \sum_{s\in [t]_0}\sum_{\substack{\ind i\in [n]^s\\\ i_1\leq i_2\leq \dots \leq i_s}}p_{\ind i} X(i_1)\dots X(i_s)}_{\op},
\end{equation}
where  $\norm{X}_{\op}$ stands for the operator norm of $X$ when regarded as a linear operator from $\ell_2(\mathbb R^m)$ to $\ell_2(\mathbb R^m).$

\subsection{Our results} In this work we prove functional inequalities involving the completely bounded norm for polynomials, and via the characterization of quantum query algorithms in terms of this norm, we obtain limiting results for different classes of quantum query algorithms. 

\subsubsection{An optimal root-influence bound  for block-multilinear polynomials}

Bansal, Sinha and de Wolf already used the characterization of quantum query algorithms in terms of polynomials whose completely bounded norm is at most one (from now on, \emph{completely bounded polynomials}) to make progress in the Aaronson--Ambainis conjecture \cite{Bansal:2022}. They focused on the special case of quantum algorithms for computing functions $$f:\underbrace{\{-1,1\}^n\times\dots\times \{-1,1\}^n}_{t\text{ times}}\to\{-1,1\},$$
where the input $x$ is divided in $t$ blocks $x_1,\dots,x_t$. They considered quantum query algorithms that make $t$ queries, one to each block, as in \cref{fig:QQAblocks}. For instance, $t$-fold forrelation, which optimally separates quantum and classical query complexity, can be encoded as an amplitude of one of these algorithms \cite{aaronson2015forrelation}. 

The amplitudes of these algorithms, and their restrictions, are completely bounded block-multilinear polynomials \cite{QQA=CBF}. More precisely, these amplitudes, and their restrictions, are polynomials $a:\{-1,1\}^n\times\dots\times \{-1,1\}^n\to\R$ \footnote{Usually, the amplitudes of quantum query algorithms are said to be complex, but we take them as real without loss of generality. This is because any quantum algorithm that makes $t$ queries can be simulated via a quantum algorithm that makes $t$ queries and only uses orthogonal matrices and computational basis measurements \cite{mckague2009simulating}.} that are block multilinear, meaning that they are affine on every block of variables, i.e. they can be written as 
\begin{equation*}
    a(x_1,\dots,x_t)=\sum_{\ind i \in [n]^t_0}a_{\ind i} x_1(i_1)\dots x_t(i_t),
\end{equation*}
for some real numbers $a_\ind i$, and where we established the convention that $x_s(0)=1$ for every $s\in [t].$
Furthermore, $\|a\|_{\cb}\leq 1$. In the case of block-multilinear polynomials, the completely bounded norm of \cref{eq:cbnorm}, when ordering the variables as $$x_1(1),\dots,x_1(n),x_2(1),\dots,x_2(n),\dots,x_t(1),\dots,x_t(n),$$ can be expressed as
\begin{equation}\label{eq:bmcbnorm}
    \norm{a}_{\cb}=\sup_{\substack{m\in\N,\, X_s(i)\in\R^{m\times m}\\ \norm{X_s(i)}_{\op}\leq 1\\ X_s(0)=\Id}} \norm{\sum_{\ind i \in [n]^t_0}a_{\ind i} X_1(i_1)\dots X_t(i_t)}_{\op}.
\end{equation}

Bansal et\ al. proved that
\begin{equation}\label{eq:BSdW}
    \maxinf [a]\norm{a}_{\cb}\geq \frac{\var[a]^2}{e(t+1)^4}.
\end{equation}

Hence, according to the bound $\|a\|_{\cb}\leq 1$ and \cref{eq:BSdW}, these amplitudes are in the conditions of \cref{con:AAconjecture}, so they can be classically and efficiently simulated as in \cref{con:needforstructure}.  

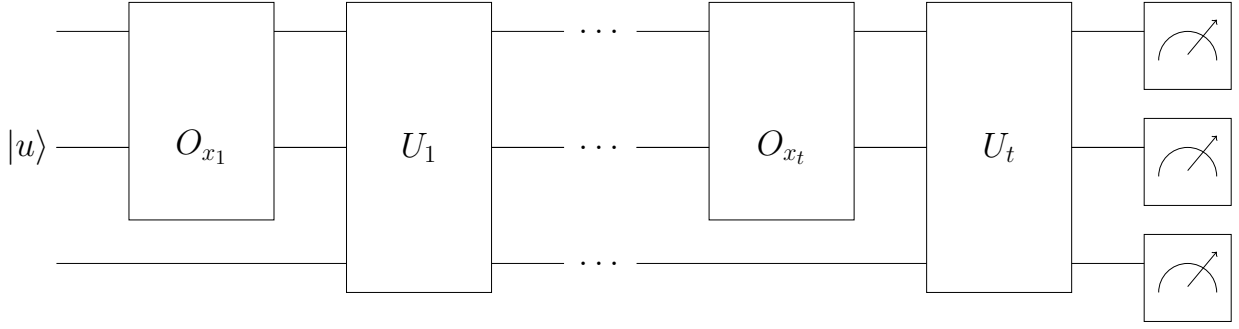
\begin{figure}[h!]
\centering
\resizebox{1\textwidth}{!}{
\begin{tikzpicture}[every node/.style={font=\LARGE}]

  \def\ycenter{13.25}

  \draw (10,15.75) rectangle (12.5,10.75) node[pos=.5] {$U_1$};
  \draw (20,15.75) rectangle (22.5,10.75) node[pos=.5] {$U_t$};

  \draw (6.25,15.75) rectangle (8.75,12);
  \node at (7.5,\ycenter) {$O_{x_1}$};

  \draw (16.25,15.75) rectangle (18.75,12);
  \node at (17.5,\ycenter) {$O_{x_t}$};

  \draw (23.75,15.75) rectangle (25.25,14.25);

  \draw (25,14.75) arc (0:180:0.5); 
  \draw[->] (24.5,14.85) -- ++(0.5,0.6); 

  \draw (23.75,13.75) rectangle (25.25,12.25);
  \draw (25,12.75) arc (0:180:0.5); 
  \draw[->] (24.5,12.85) -- ++(0.5,0.6); 

  \draw (23.75,11.75) rectangle (25.25,10.25);
  \draw (25,10.75) arc (0:180:0.5); 
  \draw[->] (24.5,10.85) -- ++(0.5,0.6); 

  \foreach \xA/\xB in {
    5/6.25, 8.75/10, 12.5/13.75,
    5/6.25, 8.75/10, 12.5/13.75,
    15/16.25, 18.75/20, 22.5/23.75,
    15/16.25, 18.75/20, 22.5/23.75
  } {
    \draw (\xA,13.25) -- (\xB,13.25);
  }

  \draw (5,15.25) -- (6.25,15.25);
  \draw (8.75,15.25) -- (10,15.25);
  \draw (12.5,15.25) -- (13.75,15.25);
  \draw (15,15.25) -- (16.25,15.25);
  \draw (18.75,15.25) -- (20,15.25);
  \draw (22.5,15.25) -- (23.75,15.25);

  \draw (5,11.25) -- (10,11.25);
  \draw (12.5,11.25) -- (13.75,11.25);
  \draw (15,11.25) -- (20,11.25);
    \draw (22.5,11.25) -- (23.75,11.25);

  \node at (4.5,13.25) {$|u\rangle$};
    \node at (14.4,11.25) {$\dots$};
  \node at (14.4,13.25) {$\dots$};
  \node at (14.4,15.25) {$\dots$};
\end{tikzpicture}

}
\caption{Quantum query algorithm that queries disjoint blocks of variables $x_1,\dots,x_t$.}
\label{fig:QQAblocks}
\end{figure}

To arrive at \cref{eq:BSdW}, Bansal et al.\ proved a stronger inequality, namely that for every homogeneous block-multilinear polynomial $a$ of degree at most $t$ and for every $s\in \{1,t\}$, one has that 
\begin{equation*}
    \norm{a}_{\cb}\geq \sum_{i\in [n]}\frac{\sqrt{\Inf_{x_s(i)}[a]}}{\sqrt{e(t+1)}}.
\end{equation*}
This was later improved by Escudero-Gutiérrez in ~\cite{gutierrez2023influences} in two ways: by showing that the optimal constant in the inequality is 1, namely 
\begin{equation*}
    \norm{a}_{\cb}\geq \sum_{i\in [n]}\sqrt{\Inf_{x_s(i)}[a]},
\end{equation*}and by showing that this inequality holds for every $s\in [t]$, not just for $s\in\{1,t\}$. 

For the general, non-homogeneous case, using that $\norm{a_{d}}_{\cb}\leq \norm{a}_{\cb}$ \cite[Eq. (22)]{gutierrez2023influences}, the above equation implies that 
\begin{equation*}
    \norm{a}_{\cb}\geq \sum_{i\in [n]}\frac{\sqrt{\Inf_{x_s(i)}[a]}}{t}
\end{equation*}for every $s\in [t]$. Summing over all blocks, we arrive at 
 \begin{equation}\label{eq:sota}
    \norm{a}_{\cb}\geq \sum_{s\in [t]}\sum_{i\in [n]}\frac{\sqrt{\Inf_{x_s(i)}[a]}}{t^2},
\end{equation}
for every block-multilinear polynomial $a$. 

Our first result consists in finding the optimal constant in \cref{eq:sota}. 

\begin{restatable}{theorem}{theoOptimalAAforBMCB}\label{theo:OptimalAAforBMCB}
    Let $a:\{-1,1\}^n\times\dots\times \{-1,1\}^n\to \R$ be a block-multilinear polynomial. Then, 
    \begin{equation*}
     \norm{a}_{\cb}\geq \sum_{s\in [t]}\sum_{i\in [n]}\frac{\sqrt{\Inf_{x_s(i)}[a]}}{t}.
\end{equation*}
    This inequality is optimal, as witnessed by $a(x_1,\dots,x_t)=x_1(1)\cdot\dots\cdot x_t(1).$
\end{restatable}

Interestingly, we will show below in \cref{cor:blockQQAsim} that~\cref{theo:OptimalAAforBMCB} implies an efficient classical simulation result that improves over the state-of-the-art results that go via the Aaronson--Ambainis argument in \cite[Theorem 22]{Aaronsons:2014}, both quantitatively, providing polynomial speedups, and qualitatively, by using non-adaptive classical queries.

However, to improve on the state-of-the-art of the simulation of quantum query algorithms we need an alternative simulation argument to the one of Aaronson and Ambainis. Indeed, by using \cref{theo:OptimalAAforBMCB}, it follows that block-multilinear polynomials satisfy
\begin{equation*}
	\norm{a}_{\cb}\geq \sum_{s\in [t]}\sum_{i\in [n]}\frac{\sqrt{\Inf_{x_s(i)}[a]}}{t}\geq \sum_{s\in [t]}\sum_{i\in [n]}\frac{\Inf_{x_s(i)}[a]}{t\sqrt{\maxinf[a]}}=\frac{\Inf[a]}{t\sqrt{\maxinf[a]}},
\end{equation*} 
where $\Inf[a]=\sum_{s\in [t]}\sum_{i\in [n]}\Inf_{x_s(i)}[a]$ is the total influence of $a$. By rearranging, we see that the amplitudes of quantum query algorithms $a$ as in \cref{fig:QQAblocks}, for which we know that $\|a\|_{\cb}\leq 1$, satisfy 
\begin{equation}\label{eq:InfAA}
	\maxinf[a]\geq \frac{\Inf^2[a]}{t^2}.
\end{equation}
By using Poincaré's inequality $\Inf[a]\geq \var[a],$ we deduce that \cref{eq:InfAA} implies 
\begin{equation}\label{eq:varAA}
	\maxinf[a]\geq \frac{\var^2[a]}{t^2},
\end{equation}
which was proven in~\cite{gutierrez2023influences} and it is the starting point of the simulation argument by Aaronson and Ambainis. To improve on that we work directly with \cref{theo:OptimalAAforBMCB} (which is stronger than \cref{eq:varAA}) and propose a new simulation argument inspired by the Friedgut junta theorem, which establishes that functions with low total influence are close to juntas~\cite{friedgut1998boolean}. The new simulation argument is as follows.

\begin{restatable}{theorem}{theosimfromsqrtinf}\label{theo:simfromsqrtinf}
	Let $p:\{-1,1\}^n\to \R$. Then, there exists a junta $q:\{-1,1\}^n\to \R$ depending on at most $(\sum_i\sqrt{\Inf_i[p]})^2/(\eps^2\delta)$ variables such that $$\E_x[|q(x)-p(x)|^2]\leq \eps^2\delta.$$
	In particular, there is a classical non-adaptive deterministic algorithm that makes $(\sum_i\sqrt{\Inf_i[p]})^2/(\eps^2\delta)$ queries and approximates $p$ up to error $\eps$ on a $(1-\delta)$-fraction of the inputs. 
\end{restatable}

The proof of \cref{theo:simfromsqrtinf} is simple. We just take $J\subseteq [n]$ to be the set of the variables with influence at least $\eps^4\delta^2/(\sum_i\sqrt{\Inf_i[p]})^2$, and take $q(x_J):=\mathbb E_{x_{[n]-J}}[p(x_J,x_{[n]-J})].$ Combining \cref{theo:OptimalAAforBMCB,theo:simfromsqrtinf}, we obtain a simulation result for quantum query algorithms as in \cref{fig:QQAblocks}.
\begin{corollary}\label{cor:blockQQAsim}
    Let $a:\{-1,1\}^n\times\dots\times\{-1,1\}^n\to \R$ be an amplitude of a $t$-query quantum algorithm that queries $t$ disjoint blocks of inputs, as in \cref{fig:QQAblocks}. Then, there is a deterministic classical algorithm  that $\eps$-approximates the amplitude on at least a ($1-\delta$)-fraction of the inputs making only $O(t^2/\eps^2\delta)$ non-adaptive queries.
\end{corollary}

\cref{cor:blockQQAsim} improves on the prior state-of-the-art simulation for quantum algorithms as in \cref{fig:QQAblocks}. This algorithm follows from the Aaronson--Ambainis theorem of~\cite{gutierrez2023influences} and the simulation argument of~\cite[Theorem 22]{Aaronsons:2014}, which together yield a classical query complexity of $O(t^3/(\eps^4\delta^3))$, and it uses adaptive queries.

\vspace{0.3cm}
\noindent\textbf{Proof idea of \cref{theo:OptimalAAforBMCB}.} We prove \cref{theo:OptimalAAforBMCB} by using a novel construction of matrices $X_s(i)$ that give a suitable lower bound to \cref{eq:bmcbnorm}. More precisely, we prove that 
\begin{equation}\label{eq:block1}
    \norm{a}_{\cb}\geq \sum_{i\in [n]_0}\sqrt{\Inf_{x_s(i)}[a]},
\end{equation}
for every $s\in [n],$ and then \cref{theo:OptimalAAforBMCB} follows by summing over the blocks. We sketch the proof of~\cref{eq:block1} for the case $s=1$, which is less notation-heavy than for the cases $1<s<t$. We construct bounded matrices $X_2(i_2),\dots,X_t(i_t)$ and a unit vector $e$ such that $$\{u_{\ind j}:=X_2(j_1)\dots X_t(j_t)e\}_{j_2,\dots,j_t\in [n]_0}$$ is an orthonormal set and also following the convention $X_s(0)=\Id$. Now, we can define $$v_i:=\sum_{\ind j\in [n]_0^{t-1}}a_{i,\ind j}u_\ind j.$$ Then, by taking $X_1(i_1)$ satisfying $X_1(i_1)v_{i_1}/\norm{v_{i_1}}=v_0/\norm{v_0}_{\ell_2}$, again with $X_1(0)=\Id$, we have that 
\begin{equation*}
    \norm{a}_{\cb}\geq \sum_{i\in [n]_0} \left\langle \frac{v_0}{\norm{v_0}_{\ell_2}},X_1(i)v_{i}\right\rangle=\sum_{i\in [n]_0}\norm{v_i}_{\ell_2}=\sum_{i\in [n]_0}\sqrt{\sum_{\ind j\in [n]^{t-1}}a_{i,\ind j}^2}.
\end{equation*}
Now, by Parseval's identity, we have that $\Inf_{x_1(i)}[a]=\sum_{\ind j\in [n]^{t-1}}a_{i,\ind j}^2$, so 
\begin{equation*}
    \norm{a}_{\cb}\geq \sum_{i\in [n]_0}\sqrt{\Inf_{x_1(i)}[a]}.
\end{equation*}

\vspace{0.3cm}
\noindent\textbf{Some comments about \cref{theo:OptimalAAforBMCB} and \cref{cor:blockQQAsim}.} We prove another simulation result, apart from \cref{theo:simfromsqrtinf}, which instead of building upon inequalities like the one in~\cref{theo:optimalcbBHBlockMult}, builds upon weaker inequalities like the one in~\cref{eq:InfAA}. Thus, together with the simulation theorem of Aaronson and Ambainis, we have three different simulation results, relying on different functional analytic assumptions. For these results, the stronger the analytic assumption, the better the simulation result. We discuss and compare the three simulation results in \cref{sec:simulationresults}.

We believe that both of these simulation results are of independent interest. For instance, in order to prove Aaronson and Ambainis conjecture with exponential dependence on the degree, it was shown in~\cite[Section 4]{defant2019fourier} that every polynomial $p:\{-1,1\}^n\to [-1,1]$ of degree $t$ satisfies $\sum_i\sqrt{\Inf_i[p]}\leq t^2e^t$. This, together with \cref{theo:simfromsqrtinf}, implies that $t$-query quantum algorithms can be simulated with $\exp(O(t))$ non-adaptive classical queries almost everywhere. Such a simulation result, with adaptive queries, was the state-of-the-art of the simulation conjecture for general bounded polynomials. Hence, \cref{theo:simfromsqrtinf} establishes a new state-of-the-art for the simulation of bounded polynomials. Sadly, \cref{theo:simfromsqrtinf} cannot lead to an efficient classical simulation for all bounded polynomials, because the address function of degree $t$ satisfies $\sum_i\sqrt{\Inf_i[p]}\geq 2^{t-2}.$

Finally, we want to mention that the proof technique of  \cref{theo:optimalcbBHBlockMult} also allows us to show that completely bounded block-multilinear polynomials satisfy a cb-Bohnenblust--Hille inequality with constant 1, as stated in \cref{coro:bhblockmultilinear}. Namely, we show that 
\begin{equation}\label{eq:BHcbbm}
    \left(\sum_{\ind i\in [n]_0^t}|a_{\ind i}|^{\frac{2d}{d+1}}\right)^{\frac{d+1}{2d}}\leq \norm{a}_{\cb}.
\end{equation}
This improves upon the previous best version of the inequality, proved in \cite{arunachalam2025cb}, which yields a constant of $\sqrt{t+1}$ for arbitrary polynomials of degree $t$.

Bohnenblust--Hille inequalities have led to major breakthroughs in the study of Dirichlet series \cite{bohnenblust1931absolute,defant2011bohnenblust}, and more recently in learning theory \cite{eskenazis2022learning,volberg2023noncommutative,slote2024bohnenblust}. In particular, \cref{eq:BHcbbm} implies an improved PAC learning result for the amplitudes of quantum query algorithms as in \cref{fig:QQAblocks} by following the ideas of \cite{arunachalam2025cb}.

\subsubsection{Optimal Fourier growth of the highest level of quantum query algorithms}
We give the first approach that gives upper bounds to the Fourier growth of quantum query algorithms via completely bounded polynomials. In particular, we focus on a recent question by Girish \cite[Question 4]{girish2026fourier}. She asked whether the dependence on $d,t$ of her upper bound to Fourier growth for quantum query algorithms could be improved. She showed that the acceptance probabilities of standard $t$-query quantum algorithms, i.e. those that query the whole input every time as in \cref{fig:QQA}, satisfy that 
\begin{equation*}\label{eq:UmasBound}
    \norm{\widehat p_{d}}_{\ell_1}\leq c^{d}\sqrt{d!}t(\log(n))^{d-1}n^{\frac{d-1}{2}},
\end{equation*}
for some constant $c$, and for every $d,t\in\N$. More generally, Iyer, Rao, Reis, Rothvoss, and Yehudayoff showed that for polynomials $p$ of degree $2t$, one has \cite{iyer2021tight}
\begin{equation*}\label{eq:IyersBound}
    \norm{\widehat p_{d}}_{\ell_1}\leq t^de^{{d+1\choose 2}} n^{\frac{d-1}{2}}.
\end{equation*}
In particular, for the acceptance probabilities of $t$-query quantum algorithms, the best prior upper bound to the Fourier growth of the highest level is 
\begin{equation}\label{eq:2tgrowth}
    \norm{\widehat p_{2t}}_{\ell_1}\leq \min\{ c^{2t}\sqrt{2t!}t(\log(n))^{d-1}n^{\frac{2t-1}{2}}, t^{2t}e^{{2t+1\choose 2}} n^{\frac{2t-1}{2}}\}.
\end{equation}

Our second result is a tight Fourier growth bound for the highest level of $t$-query quantum algorithms, solving the highest level case of the question by Girish. 

\begin{restatable}{theorem}{theooptFGstd}\label{theo:optFGstd}
    Let $p:\{-1,1\}^n\to \mathbb{R}$ be the acceptance probability of a $t$-query quantum algorithm as in \cref{fig:QQA}. Then, 
   \begin{equation*}
        \norm{\widehat p_{2t}}_{\ell_1}\leq \sqrt{{n-1\choose 2t-1}}.
\end{equation*}  
\end{restatable}

In particular, from \cref{theo:optFGstd}, it follows that for $t$-query quantum algorithms, $\norm{\widehat p_{2t}}_{\ell_1}\leq (en/(2t-1))^{t-1/2}$. This is nearly tight, as $2t$-fold forrelation can be computed with $t$ queries, and satisfies $\norm{\widehat p_{2t}}_{\ell_1}=(n/2t)^{t-1/2}$.

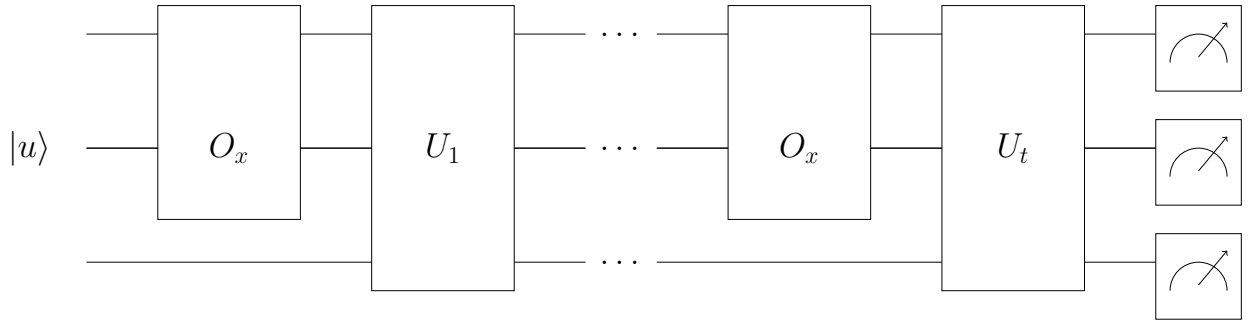
\begin{figure}[ht!]
\centering
\resizebox{1\textwidth}{!}{%
\begin{tikzpicture}[every node/.style={font=\LARGE}]

  \def\ycenter{13.25}

  \draw (10,15.75) rectangle (12.5,10.75) node[pos=.5] {$U_1$};
  \draw (20,15.75) rectangle (22.5,10.75) node[pos=.5] {$U_t$};

  \draw (6.25,15.75) rectangle (8.75,12);
  \node at (7.5,\ycenter) {$O_x$};

  \draw (16.25,15.75) rectangle (18.75,12);
  \node at (17.5,\ycenter) {$O_x$};

  \draw (23.75,15.75) rectangle (25.25,14.25);

  \draw (25,14.75) arc (0:180:0.5); 
  \draw[->] (24.5,14.85) -- ++(0.5,0.6); 

  \draw (23.75,13.75) rectangle (25.25,12.25);
  \draw (25,12.75) arc (0:180:0.5); 
  \draw[->] (24.5,12.85) -- ++(0.5,0.6); 

  \draw (23.75,11.75) rectangle (25.25,10.25);
  \draw (25,10.75) arc (0:180:0.5); 
  \draw[->] (24.5,10.85) -- ++(0.5,0.6); 

  \foreach \xA/\xB in {
    5/6.25, 8.75/10, 12.5/13.75,
    5/6.25, 8.75/10, 12.5/13.75,
    15/16.25, 18.75/20, 22.5/23.75,
    15/16.25, 18.75/20, 22.5/23.75
  } {
    \draw (\xA,13.25) -- (\xB,13.25);
  }

  \draw (5,15.25) -- (6.25,15.25);
  \draw (8.75,15.25) -- (10,15.25);
  \draw (12.5,15.25) -- (13.75,15.25);
  \draw (15,15.25) -- (16.25,15.25);
  \draw (18.75,15.25) -- (20,15.25);
  \draw (22.5,15.25) -- (23.75,15.25);

  \draw (5,11.25) -- (10,11.25);
  \draw (12.5,11.25) -- (13.75,11.25);
  \draw (15,11.25) -- (20,11.25);
    \draw (22.5,11.25) -- (23.75,11.25);

  \node at (4,13.25) {$|u\rangle$};
  \node at (14.4,11.25) {$\dots$};
  \node at (14.4,13.25) {$\dots$};
  \node at (14.4,15.25) {$\dots$};

\end{tikzpicture}

}
\caption{Standard quantum query algorithm.}
\label{fig:QQA}
\end{figure}

\noindent\textbf{Sketch of the proof.} We build upon \cite{QQA=CBF}, where it was shown that given the acceptance probability $p:\{-1,1\}^n\to \R$ of a $t$-query quantum algorithm as in \cref{fig:QQA}, there is a $2t$-linear form $T:\mathbb R^{2n}\times\dots\times\R^{2n}\to \R$ which
    \begin{enumerate}
        \item extends $p,$ meaning that $p(x)=T((x,1^n), \dots,(x,1^n))$  for every $x\in\{-1,1\}^n$, \label{item:1intro}
        \item is completely bounded, meaning that its completely bounded norm, $$ \norm{T}_{\cb}=\sup_{\substack{m\in\N,\, X_s(i)\in\R^{m\times m}\\ \norm{X_s(i)}_{\op}\leq 1}}\norm{\sum_{\ind i\in [2n]^{2t}}T_{\ind i}X_1(i_1)\dots X_{2t}(i_{2t})}_{\mathrm{op}},$$
        is smaller than 1.\label{item:2intro}
    \end{enumerate}
Our main contribution towards \cref{theo:optFGstd} is to find an explicit tuple of matrices $X(i)$ and vectors $e_\emptyset$ and $v$, that satisfy 
\begin{align}
        \langle e_\emptyset, X(i_1)\dots X(i_{2t}) v\rangle&=\left\{\begin{array}{ll}
           \frac{\sign{(\widehat p (S_{\ind i}))}}{\sqrt{{n-1\choose 2t-1}}}  &  \text{if } |S_\ind i|=2t,\\
           0  & \text{otherwise}.
        \end{array}\right.\label{eq:keyeqFG}
\end{align}
Here, given $\ind i\in [2n]^{2t}$ we denote $S_\ind i=\{i\in [n]:\, i \text{ occurs an odd number of times in }\ind i\}.$ Once there,
\begin{equation*}
    \norm{T}_{\cb}\geq \frac{1}{\sqrt{{n-1\choose 2t-1}}}\sum_{\ind i\in [2n]^t}\delta_{|S_\ind i|=2t}T_{\ind i}\sign((\widehat p(S_\ind i))=\frac{1}{\sqrt{{n-1\choose 2t-1}}}\sum_{|S|=2t}\sign((\widehat p(S))\sum_{\ind i\in [2n]^t:\ S_\ind i=S}T_{\ind i}.
\end{equation*}
The next step of the proof is to notice that \cref{item:1intro} implies $\widehat p(S)=\sum_{\ind i\in [2n]^t:\ S_\ind i=S}T_{\ind i},$ which leads to the functional inequality
\begin{equation*}
    \norm{T}_{\cb}\geq \frac{1}{\sqrt{{n-1\choose 2t-1}}}\sum_{|S|=2t}\sign (\widehat p(S)) \widehat p(S)=\frac{1}{\sqrt{{n-1\choose 2t-1}}}\sum_{|S|=2t} |\widehat p(S)|.
\end{equation*}
Finally, by using \cref{item:2intro} one concludes the proof.

\vspace{0.3cm}
\noindent\textbf{Some comments about \cref{theo:optFGstd}.} Sadly our proof technique only works for the highest and the second highest levels. For the second highest level we can show $\norm{\widehat p_{2t-1}}_{\ell_1}\leq \sqrt{{n-1\choose 2t-2}},$ which also improves over prior works by a large factor. However, for the levels $2t-s$ for $s\geq 2,$ our matrices no longer satisfy \cref{eq:keyeqFG}, as we detail in \cref{rem:wherebreaks}. 

Nonetheless, we can prove a generalization of \cref{theo:optFGstd} to quantum query algorithms with bounded adaptivity (see \cref{prop:fgrowthadaptivequeriesfgrowth}). If one allows for $r-1$ rounds of adaptivity, as in \cref{fig:QQAboundedadativity}, we can show that the acceptance probability of a $t$-query quantum algorithm satisfies 
\begin{equation}\label{eq:oursFGadapt}
    \norm{\widehat p_{2t}}_{\ell_1}\leq \left(\frac{2en}{2t-t_{\max}}\right)^{t-t_{\max}/2},
\end{equation}
where $t_{\max}$ is the maximum number of parallel queries made by the algorithm. The previous best bound was by Girish, Sinha, Tal and  Wu~\cite[Corollary 4.2]{girish2024power}, who showed that 
\begin{equation}\label{eq:previousFGadapt}
    \norm{\widehat p_{2t}}_{\ell_1}\leq 2^{rt\min\{2^{4rt,(2t)^{4^r}}\}}t^{2t}\left(\frac{n}{t}\right)^{t-\frac{t}{2r}}.
\end{equation}
Note that the worst case of our bound in \cref{eq:oursFGadapt} is when $t_{\max}=t/r.$ Then, one recovers the dependence on $n$ in \cref{eq:previousFGadapt}, and improves \cref{eq:previousFGadapt}.
Finally, we also want to stress that, in contrast to prior work, our bound of \cref{eq:oursFGadapt} depends on $n$ as 
\begin{equation*}
    n^{t-t_{\max}/2},
\end{equation*}
which decreases when $t_{\max}$ increases. Thus, our bound not only suggests that bounded adaptivity limits the power of quantum query algorithms, but it goes further and quantifies the limitations caused by making many queries in parallel.

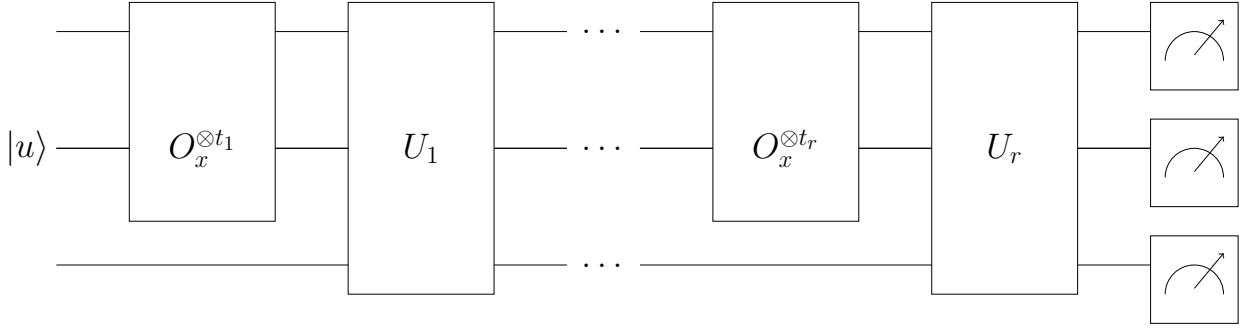
\begin{figure}[h!]
\centering
\resizebox{1\textwidth}{!}{%
\begin{tikzpicture}[every node/.style={font=\LARGE}]

  \def\ycenter{13.25}

  \draw (10,15.75) rectangle (12.5,10.75) node[pos=.5] {$U_1$};
  \draw (20,15.75) rectangle (22.5,10.75) node[pos=.5] {$U_r$};

  \draw (6.25,15.75) rectangle (8.75,12);
  \node at (7.5,\ycenter) {$O_{x}^{\otimes t_1}$};

  \draw (16.25,15.75) rectangle (18.75,12);
  \node at (17.5,\ycenter) {$O_{x}^{\otimes t_r}$};

  \draw (23.75,15.75) rectangle (25.25,14.25);

  \draw (25,14.75) arc (0:180:0.5); 
  \draw[->] (24.5,14.85) -- ++(0.5,0.6); 

  \draw (23.75,13.75) rectangle (25.25,12.25);
  \draw (25,12.75) arc (0:180:0.5); 
  \draw[->] (24.5,12.85) -- ++(0.5,0.6); 

  \draw (23.75,11.75) rectangle (25.25,10.25);
  \draw (25,10.75) arc (0:180:0.5); 
  \draw[->] (24.5,10.85) -- ++(0.5,0.6); 

  \foreach \xA/\xB in {
    5/6.25, 8.75/10, 12.5/13.75,
    5/6.25, 8.75/10, 12.5/13.75,
    15/16.25, 18.75/20, 22.5/23.75,
    15/16.25, 18.75/20, 22.5/23.75
  } {
    \draw (\xA,13.25) -- (\xB,13.25);
  }

  \draw (5,15.25) -- (6.25,15.25);
  \draw (8.75,15.25) -- (10,15.25);
  \draw (12.5,15.25) -- (13.75,15.25);
  \draw (15,15.25) -- (16.25,15.25);
  \draw (18.75,15.25) -- (20,15.25);
  \draw (22.5,15.25) -- (23.75,15.25);

  \draw (5,11.25) -- (10,11.25);
  \draw (12.5,11.25) -- (13.75,11.25);
  \draw (15,11.25) -- (20,11.25);
    \draw (22.5,11.25) -- (23.75,11.25);

  \node at (4.5,13.25) {$|u\rangle$};
  \node at (14.4,11.25) {$\dots$};
  \node at (14.4,13.25) {$\dots$};
  \node at (14.4,15.25) {$\dots$};
\end{tikzpicture}

}%
\caption{Quantum query algorithm with $r-1$ rounds of adaptivity.}
\label{fig:QQAboundedadativity}
\end{figure}

\subsubsection{Classical simulation of non-adaptive quantum query algorithms}
We also consider non-adaptive quantum query algorithms that make the $t$ queries in parallel, as in \cref{fig:QQAnonadaptive}. We make a simple, yet powerful, remark that goes beyond the standard polynomial method: the amplitudes of these algorithms are not only bounded polynomials of degree $t,$ but they are restrictions of bounded polynomials of degree just 1. As a consequence, we can prove the following. 

\begin{restatable}{proposition}{propell}\label{prop:ell1boundamplitudes}
    Let $a:\{-1,1\}^n\to\R$ be the amplitude of a non-adaptive quantum algorithm that makes $t$ queries, as in \cref{fig:QQAnonadaptive}. Then, 
    \begin{equation*}
        \norm{\widehat a}_{\ell_1}\leq 1.
    \end{equation*}
\end{restatable}

Note that the standard polynomial method states that the amplitudes of these algorithms are bounded and of degree at most $t,$ which only implies that $\norm{\widehat a}_{\ell_1}=O_{t}(n^{(t-1)/2})$~\cite{iyer2021tight}. We combine \cref{prop:ell1boundamplitudes} with a folklore simulation technique for polynomials with bounded $\ell_1$-norm \cite{grolmusz1997power}. This lets us show that quantum query algorithms which measure all qubits and only accept on a few measurement outcomes (for which the acceptance probability is the sum of the squares of a few amplitudes) can be efficiently and classically simulated. 

\begin{restatable}{corollary}{corsimnonadapt}\label{cor:simnonadapt}
    Let $p:\{-1,1\}^n\to [0,1]$ be the acceptance probability of a $t$-query non-adaptive quantum algorithm that measures all qubits and accepts only on $r$ of the measurement outcomes. Then, there is a non-adaptive randomized classical algorithm that makes only $$O\left(\frac{tr^2}{\eps^2}\cdot\log\left(\frac{1}{\delta}\right)\right)$$ queries, and that for every input $x\in\{-1,1\}^n$, outputs an estimate $\widetilde p(x)$ that, with probability $\geq1-\delta$, satisfies $|\widetilde p(x)-p(x)|\leq \eps.$ 
\end{restatable}

The kind of simulation of \cref{cor:simnonadapt} is incomparable with the one of \cref{con:needforstructure}, as we only use parallel classical queries and simulate the algorithm everywhere, but we only succeed with high probability.

\vspace{0.3cm}
\noindent\textbf{Sketch of the proof.} Given a polynomial $p:\{-1,1\}^n\to \R,$ by the triangle inequality one has that $\norm{p}_{\cb}= \norm{\widehat p}_{\ell_1}$. If the degree of the polynomial is $1,$ then one has that $\norm{p}_{\cb}\leq \norm{\widehat p}_{\ell_1}$, which is an optimal functional inequality. Using this, and the remark that the amplitudes of quantum query algorithms as in \cref{fig:QQAnonadaptive} are restrictions of degree 1 polynomials, we prove \cref{prop:ell1boundamplitudes}. 

To prove \cref{cor:simnonadapt}, we approximate $p(x)$, for a polynomial $p:\{-1,1\}^n\to \R$ of degree $t$  $\norm{\widehat p}_{\ell_1}\leq r$, and an element $x\in\{-1,1\}^n$. To do that, we sample $S_1,\dots,S_T$ from the (renormalized) distribution defined by the absolute values of the Fourier coefficients, and define 
\begin{equation*}
    \widetilde p(x)=\sum_{i\in [T]}\widehat p(S_i)\chi_{S_i}(x).
\end{equation*}
We have that $\mathbb E[\widetilde p(x)]=p(x)$. Thus, thanks to the Hoeffding bound, we have that $|\widetilde p(x)-p(x)|\leq \eps$ with high probability when $T=O(\frac{r^2}{\eps^2}\log(\frac{1}{\delta}))$. Finally, we note that one can evaluate $\widetilde p(x)$ by only querying $tT=O(t\frac{r^2}{\eps^2}\log(\frac{1}{\delta}))$ entries of $x.$ Therefore, the total query complexity of the process is $tT=O(t\frac{r^2}{\eps^2}\log(\frac{1}{\delta})),$ as claimed. 

\vspace{0.3cm}
\noindent\textbf{Comments about \cref{prop:ell1boundamplitudes,cor:simnonadapt}.} \cref{prop:ell1boundamplitudes} shows a striking separation between quantum algorithms with and without adaptivity. This is because quantum query algorithms with just 2 queries and one round of adaptivity can encode forrelation in an amplitude. Hence, with a single round of adaptivity, there are amplitudes with $\norm{\widehat a}_{\ell_1}=(n/2)^{1/2}$, which sharply contrasts with the amplitudes of non-adaptive quantum query algorithms, that satisfy $\norm{\widehat a}_{\ell_1}\leq 1$ by \cref{prop:ell1boundamplitudes}.

Regarding \cref{cor:simnonadapt}, during the preparation of this manuscript we became aware of two other works that prove related results \cite{AA2026boundedadaptivity,liu2026parallel}.
Blanc, Docter, Strassle, and Tan showed that if $p$ is the acceptance probability of a quantum query algorithm with $r-1$ rounds of adaptivity that makes $t$ parallel queries in each query round, then there is a classical algorithm that makes $(t \log(1/\delta)/\eps)^{O(r)}$ queries and $\eps$-approximates $f$ on a $(1-\delta)$-fraction of inputs~\cite{AA2026boundedadaptivity}. Liu and Mutreja show the same result, but with $(rt/(\eps\sqrt{\delta}))^{O(4^r)}$ classical queries~\cite{liu2026parallel}. Both of these results imply efficient classical simulation for quantum query algorithms with a constant number of rounds of adaptivity, while our \cref{cor:simnonadapt} only works for non-adaptive quantum query algorithms. However, we have decided to include our result because of the simplicity of our proof and because it simulates with non-adaptive classical queries.\footnote{\cite{AA2026boundedadaptivity} communicates that in a future work they will also make their simulation algorithm use non-adaptive classical queries for the case of non-adaptive quantum query algorithms.} 
\begin{figure}[h!]
\centering
\hspace{-4cm}
\begin{tikzpicture}[every node/.style={font=\LARGE}]

  \def\ycenter{13.25}


  \draw (10,15.75) rectangle (12.5,10.75) node[pos=.5] {$U_1$};

  \draw (6.25,15.75) rectangle (8.75,12);
  \node at (7.5,\ycenter) {$O_{x}^{\otimes t}$};

  \draw (13.75,15.75) rectangle (15.25,14.25);

  \draw (15,14.75) arc (0:180:0.5); 
  \draw[->] (14.5,14.85) -- ++(0.5,0.6); 

  \draw (13.75,13.75) rectangle (15.25,12.25);
  \draw (15,12.75) arc (0:180:0.5); 
  \draw[->] (14.5,12.85) -- ++(0.5,0.6); 

  \draw (13.75,11.75) rectangle (15.25,10.25);
  \draw (15,10.75) arc (0:180:0.5); 
  \draw[->] (14.5,10.85) -- ++(0.5,0.6); 

  \foreach \xA/\xB in {
    5/6.25, 8.75/10, 12.5/13.75,
    5/6.25, 8.75/10, 12.5/13.75,
  } {
    \draw (\xA,13.25) -- (\xB,13.25);
  }

  \draw (5,15.25) -- (6.25,15.25);
  \draw (8.75,15.25) -- (10,15.25);
  \draw (12.5,15.25) -- (13.75,15.25);

  \draw (5,11.25) -- (10,11.25);
  \draw (12.5,11.25) -- (13.75,11.25);

  \node at (4.5,13.25) {$|u\rangle$};

\end{tikzpicture}

\caption{Non-adaptive quantum query algorithms.}
\label{fig:QQAnonadaptive}
\end{figure}
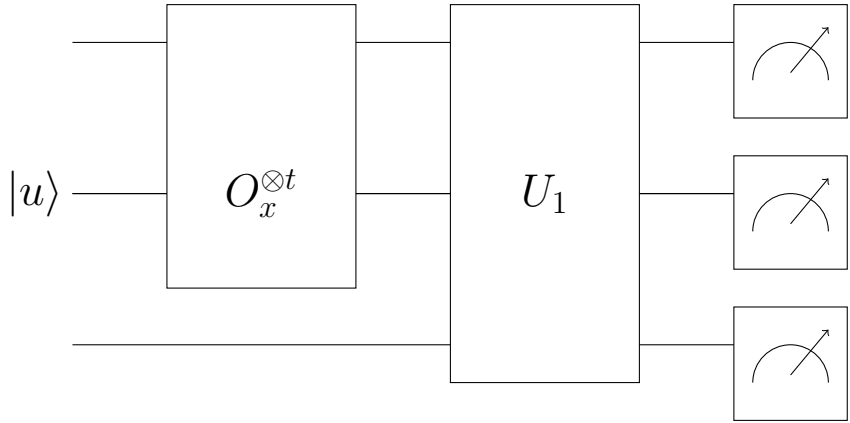

\section{Preliminaries}

\textbf{Notation. } Throughout, we use the following notation. We write
\[
[n]=\{1,\ldots,n\},\qquad
[n]_0=\{0\}\cup[n].
\]
For $x\in\mathbb{R}^n$, let
\[
O_x=\operatorname{Diag}(x,1^n),
\]
where
\[
1^n=(1,\ldots,1)\in\mathbb{R}^n.
\]
More generally, for $x\in\mathbb{R}$, we write
\[
x^n=(x,\ldots,x)\in\mathbb{R}^n.
\]
The operator norm of a linear operator $A$ is
\[
\|A\|_{\op}
=
\sup_{\|v\|_{\ell_2}=1}\|Av\|_{\ell_2}.
\]
For a polynomial $p:\R^n\to \R$, its completely bounded norm is given by 
\begin{equation}\label{eq:cbnorm2}
    \norm{p}_{\cb}:=\sup_{\substack{m\in\N,\, X(i)\in\R^{m\times m}\\ \norm{X(i)}_{\op}\leq 1}} \norm{\sum_{s\in [t]_0}\sum_{\substack{\ind i\in [n]^s\\\ i_1\leq i_2\leq \dots \leq i_s}}p_{\ind i} X(i_1)\dots X(i_s)}_{\op},
\end{equation}
Finally, for $x\in\mathbb{R}^n$ and $t\in\mathbb{N}$, we define $x^{\otimes t}$ to be the diagonal of the matrix $\operatorname{Diag}(x)^{\otimes t}$; and for a multi-index $\ind i\in[n]^t$, we write
\[
S_{\ind i}
=
\{j\in[n]:\,j\text{ occurs an odd number of times in }\ind i\}.
\]

\textbf{Fourier analysis.} We next recall some basic facts about Fourier analysis on the Boolean cube. For every subset $S\subseteq[n]$, define the associated \emph{character} by
\[
\chi_S(x)=\prod_{i\in S}x_i,\qquad x\in\{-1,1\}^n.
\]
The characters form an orthonormal basis with respect to the uniform probability measure on the Boolean cube:
\[
\mathbb{E}_x[\chi_S(x)\chi_{S'}(x)]
=
\delta_{S,S'},
\]
where
\[
\mathbb{E}_x[f(x)]
=
\frac{1}{2^n}
\sum_{x\in\{-1,1\}^n}f(x).
\]
Every function $p:\{-1,1\}^n\to\mathbb{R}$ admits the Fourier expansion
\[
p(x)
=
\sum_{S\subseteq[n]}
\widehat p(S)\chi_S(x),
\]
where
\[
\widehat p(S)
=
\mathbb{E}_x[p(x)\chi_S(x)].
\]
Parseval's identity states that 
\[
\mathbb E_x[|p(x)|^2]=\sum_{S\subseteq [n]}|\widehat p(S)|^2.
\]
The variance of $p$ is
\[
\operatorname{Var}[p]
=
\mathbb{E}_x\!\left[(p(x)-\mathbb{E}_x[p(x)])^2\right]
=
\sum_{\emptyset\neq S\subseteq[n]}|\widehat p(S)|^2.
\]
The influence of the $i$th variable is
\[
\operatorname{Inf}_i[p]
=
\frac14\mathbb{E}_x\!\left[(p(x)-p(x^{\oplus i}))^2\right]
=
\sum_{i\ni S}|\widehat p(S)|^2,
\]
where $x^{\oplus i}$ is obtained from $x$ by flipping its $i$th coordinate. The maximum influence of $p$ is
\[
\operatorname{MaxInf}[p]
=
\max_{i\in[n]}\operatorname{Inf}_i[p].
\]
For functions $p:\{-1,1\}^n\to\R$ defined on the Boolean cube, we also consider the norms
\[
\|p\|_{\infty}
=
\max_{x\in\{-1,1\}^n}|p(x)|, \quad  \norm{\widehat p}_{\ell_q}=\Big(\sum_{S\subseteq [n]}|\widehat p(S)|^q\Big)^{1/q}.
\]
For these polynomials, the completely bounded norm reads as \[
\|p\|_{\cb}
=
\sup_{\|X(i)\|_{\op}\le1,\;i\in[n]}
\left\|
\sum_{\substack{S=\{s_1,\ldots,s_t\}\subseteq[n]\\
s_1<\cdots<s_t}}
\widehat p(S)\,
X(s_1)\cdots X(s_t)
\right\|_{\op}.
\]

\textbf{Multilinear forms and block-multilinear polynomials.} A $t$-linear form is a function
\[
T:\prod_{i=1}^t\{-1,1\}^n\to\mathbb{R}
\]
which is linear with respect to each factor. Every such function can be written as
\[
T(x_1,\ldots,x_t)
=
\sum_{\ind i\in[n]^t}
T_{\ind i}\,
x_1(i_1)\cdots x_t(i_t),
\]
where $T_{\ind i}\in\mathbb{R}$ for every $\ind i\in [n]^t$. The completely bounded norm of a form $T$ is the completely bounded norm of \cref{eq:cbnorm2} when the variables are ordered as $$x_1(1),\dots,x_1(n),x_2(1),\dots,x_2(n),\dots,x_t(1),\dots,x_t(n).$$
In other words, 
\[
\|T\|_{\cb}
=
\sup_{\|X_r(j)\|_{\op}\le1,\;r\in[t],\,j\in[n]}
\left\|
\sum_{\ind i\in[n]^t}
T_{\ind i}\,
X_1(i_1)\cdots X_t(i_t)
\right\|_{\op}.
\]

A function
\[
p:\prod_{r=1}^t\{-1,1\}^n\to\mathbb{R}
\]
is said to be block multilinear if
\[
p(x_1,\ldots,x_t)
=
\sum_{\ind i\in[n]_0^t}
p_{\ind i}\,
x_1(i_1)\cdots x_t(i_t),
\]
where we adopt the convention that $x_r(0)=1$ for every $r\in[t]$.

Note that the variance and the influence of the variable $x_s(i)$ (for $s\in [t]$ and $i\in [n]$) on $p$ are respectively 
\begin{align}\label{inf_block_mult}
\operatorname{Var}[p]=\sum_{\substack{\ind i\in[n]_0^t\\ \ind i\neq (0\cdots0)}}
p_{\ind i}^2,\hspace{0.5 cm} \operatorname{Inf}_{x_s(i)}[p]=\sum_{\substack{\ind i\in[n]_0^t\\ i_s=i}}
p_{\ind i}^2.
\end{align}

The completely bounded norm of a block-multilinear polynomial $p$ is the completely bounded norm of \cref{eq:cbnorm2} when the variables are ordered as $$x_1(1),\dots,x_1(n),x_2(1),\dots,x_2(n),\dots,x_t(1),\dots,x_t(n).$$
Thus, the completely bounded norm of these polynomials can be expressed as
\[
\|p\|_{\cb}
=
\sup_{\substack{\|X_r(i)\|_{\op}\le1\\ X_r(0)=\Id}}
\left\|
\sum_{\ind i\in[n]_0^t}
p_{\ind i}\,
X_1(i_1)\cdots X_t(i_t)
\right\|_{\op}.
\]

\textbf{Algorithms.} We say that a classical query algorithm is non-adaptive if the choices of what queries to make do not depend on each other; otherwise it is adaptive. We say that an algorithm is deterministic if it does not use randomness as resource; otherwise it is randomized.

\subsection{Refinements of the polynomial method}
In this section, we review the refinement of the polynomial method for quantum query algorithms in terms of completely bounded polynomials from \cite{QQA=CBF}. We also propose similar refinements of the polynomial method for quantum query algorithms with bounded adaptivity that may be of independent interest.

\subsubsection{Standard quantum query algorithms} Here, we consider standard quantum query algorithms, as in \cref{fig:QQA}. In this case, the acceptance probability (the probability of outputting 1) can be written as 
\begin{equation}\label{eq:QQA}
    p(x)=\langle u, (O_x\otimes \Id_d)A_{1}\dots A_{2t-1} (O_x\otimes \Id_d) u\rangle,
\end{equation}
for some $n,d\in\mathbb N$, some unit vector $u$ of dimension $2nd$,  every $x \in \{-1,1\}^n$, and some square matrices  $A_1,\dots, A_{2t-1}$ of dimension $2nd$ with operator norm at most~1.

With respect to \cref{fig:QQA}, the matrices $A_s^T$ and $A_{2s-i}$ correspond to $U_s$ for $1\leq s\leq t-1$, and $A_t=U_t^T\Pi_1U_t$ where $\Pi_1$ is the projector of the final measurement $\{\Pi_0,\Pi_1\}$ that determines the probability of outputting 1.

For these algorithms, Arunachalam, Briët and Palazuelos proved the following refinement of the polynomial method \cite{QQA=CBF}. We include its proof for completeness. 
\begin{theorem}\label{theo:cbmethodstd}
    Let $p:\{-1,1\}^n\to [0,1]$ be the acceptance probability of a quantum algorithm that makes $t$ queries. Then, there exists a $2t$-linear form $T:\mathbb R^{2n}\times\dots\times\R^{2n}\to \R$ such that 
    \begin{enumerate}
        \item $p(x)=T((x,1^n), \dots,(x,1^n))$  for every $x\in\{-1,1\}^n$, \label{item:QQAstd1}
        \item $\norm{T}_{\cb}\leq 1$. \label{item:QQA2}
    \end{enumerate}
\end{theorem}
\begin{proof}
        Following the notation of \cref{eq:QQA}, the tuple $(u,A_1,\dots,A_{2t-1})$ defines a $2t$-linear form 
        $$T:\underbrace{\mathbb R^{2n}\otimes\dots\otimes \mathbb R^{2n}}_{2t \text{ times}}\to \mathbb R$$ via 
    \begin{equation*}
        T(x_1,\dots,x_{2t})=\langle u,(\Diag(x_1)\otimes \Id_d)A_{1}\dots A_{2t-1} (\Diag(x_{2t})\otimes \Id_d)u\rangle.
    \end{equation*}
   We write $T$ as
    \begin{equation*}
        T(x_1,\dots,x_{2t})=\sum_{\ind i\in [2n]^{2t}}T_{\ind i}x_1(i_1)\dots x_{2t}(i_{2t}).
    \end{equation*}
        By the definition of $T$, \cref{item:QQAstd1} is satisfied. Furthermore, we note that for matrices $X_s(i)$ of operator norm at most~1, we have that 
        \begin{align*}
            &\sum_{\ind i\in [2n]^{2t}}T_{\ind i}X_1(i_1)\dots X_{2t}(i_{2t})\\
            &=(\langle u|\otimes \Id_m ), (\Diag(X_1(1),\dots,X_1(2n))\otimes \Id_d)(A_1 \otimes \Id _m) \dots\dots (|u\rangle\otimes \Id_m).
        \end{align*}

        Finally, as all of the factors in the product of the bottom line of the last equation have operator norm at most~1, \cref{item:QQA2} follows.
    \end{proof}

\subsubsection{Quantum algorithms that query disjoint sets of inputs.} In this section we consider quantum query algorithms for computing functions $f:\{-1,1\}^n\times\dots\times  \{-1,1\}^n\to \{-1,1\}$, where the input is divided in $t$ disjoint blocks of variables 
. In this case, the quantum query algorithm makes $t$ queries, one to each block of inputs, as in \cref{fig:QQAblocks}. For instance, $t$-fold forrelation can be phrased as an amplitude of one of these algorithms \cite{aaronson2015forrelation}. The amplitudes of the state prepared by these algorithms can be written as
\begin{equation}
    a(x_1,\dots,x_t)=\langle u, (O_{x_1}\otimes \Id_d)A_{1}\dots A_{t-1} (O_{x_t}\otimes \Id_d) A_tv\rangle,
\end{equation}
for some $d\in\mathbb N$, $x_j \in \{-1,1\}^n$, some unit vectors $u$ and $v$ of dimension $2nd$, and some square matrices  $A_1,\dots, A_{t}$ of dimension $2nd$ with operator norm at most~1. With respect to \cref{fig:QQAblocks}, the matrices $A_s$ correspond to $U_s$ for $1\leq s\leq t$, and $\ket v$ is a vector of the final measurement. The acceptance probabilities of these algorithms can be written as a sum of the squares of such amplitudes.

For these algorithms, Bansal, Sinha and de Wolf observed, inspired by the above result, the following refinement of the polynomial method \cite{Bansal:2022}. Its proof is essentially the same as the one of \cref{theo:cbmethodstd}, so we omit it. 

\begin{proposition}\label{prop:cbmethodblock}
    Let $a:\{-1,1\}^n\times\dots\times  \{-1,1\}^n\to \mathbb R$ be the amplitude of a quantum algorithm that makes $t$ queries to disjoint blocks. Then, there exists a $t$-linear form $T:\mathbb R^{2n}\times\dots\times\R^{2n}\to \R$ such that 
    \begin{enumerate}
        \item $a(x_1,\dots,x_t)=T((x_1,1^n),\dots,(x_t,1^n))$ for every $x_1,\dots,x_t\in\{-1,1\}^n$,
        \item $\norm{T}_{\cb}\leq 1$.
    \end{enumerate}

\end{proposition}

\subsubsection{Quantum query algorithms with bounded adaptivity} In this section we propose a refinement of the polynomial method for quantum query algorithms that have at most $r-1$ rounds of adaptivity, as in \cref{fig:QQAboundedadativity}. Let $t_1,\dots,t_r$ be the number of queries that these algorithms make in $t$-rounds, and let $t=\sum_{i\in [r]}t_i$. In this case, the acceptance probability can be written as 
\begin{equation}\label{eq:QQAboundedadap}
    p(x)=\langle u, (O_x^{\otimes t_1}\otimes \Id_{d_1})A_{1}\dots  (O_x^{\otimes t_r}\otimes \Id_{d_r}) A_r (O_x^{\otimes t_r}\otimes \Id_{d_r})  \dots A_{2r-1} (O_x^{\otimes t_1}\otimes \Id_{d_1}) u\rangle,
\end{equation}
for some $d,d_1,\dots,d_r\in\mathbb N$ satisfying $(2n)^{t_s}d_s=d$ for all $s\in [r]$, some unit vector $u$ of dimension $d$, and some square matrices  $A_1,\dots, A_{2t-1}$ of dimension $d$ with operator norm at most~1. With respect to \cref{fig:QQAboundedadativity}, the matrices $A_s$ and $A_{2s-i}$ correspond to $U_s$ for $1\leq s\leq r-1$, and $A_t=U_r^T\Pi_1U_r$ where $\Pi_1$ is the projector of the final measurement $\{\Pi_0,\Pi_1\}$ that determines the probability of outputting 1. 

For these algorithms, we provide a new refinement of the polynomial method, inspired by \cref{theo:cbmethodstd}. The proof is very similar to the one of \cref{theo:cbmethodstd} so we omit it.

\begin{proposition}\label{theo:cbmethodboundedadapt}
    Let $p:\{-1,1\}^n\to [0,1]$ be the acceptance probability of a quantum algorithm that makes $t_1,\dots,t_r$ queries with $r-1$ rounds of adaptivity. Let $t=\sum_{i\in [r]}t_i$. Then, there exists a $2r$-linear form $T:\mathbb R^{2nt_1}\times\dots\times \mathbb R^{2nt_r}\times \mathbb R^{2nt_r}\times\dots\times\R^{2nt_1}\to \R$ such that 
    \begin{enumerate}
        \item $p(x)=T((x,1^n)^{\otimes t_1},\dots,(x,1^n)^{\otimes t_r},(x,1^n)^{\otimes t_r},\dots,(x,1^n)^{\otimes t_1})$ for every $x\in\{-1,1\}^n$, 
        
        \item $\norm{T}_{\cb}\leq 1$.
        
    \end{enumerate}
\end{proposition}

\subsubsection{Non-adaptive quantum query algorithms} In this section, we propose a refinement of the polynomial method for non-adaptive quantum algorithms, as in \cref{fig:QQAnonadaptive}. By contrast with the case of bounded adaptivity, for our later applications we will focus on the amplitudes, and not the acceptance probability. For that reason, we state it on its own, but we omit the proof of the correspondent refinement to avoid redundancy. 
Then, the amplitudes of the state prepared by the algorithm can be written as
\begin{equation}
    a(x_1,\dots,x_t)=\langle u,(O_{x}^{\otimes t}\otimes \Id_d)Av\rangle,
\end{equation} 
for some $d\in\mathbb N$, some unit vectors $u$ and $v$ of dimension $2nd$, and a square matrix $A$ of dimension $2nd$ with operator norm at most~1. With respect to \cref{fig:QQAblocks}, the matrix $A$ corresponds to $U$, and $\ket v$ is a vector of the final measurement. Here, $A$ can be absorbed into $v$. The acceptance probabilities of these algorithms can be written as a sum of the squares of such amplitudes.

\begin{proposition}\label{prop:cbmethodnonadaptive}
    Let $a:\{-1,1\}^n\to \mathbb R$ be the amplitude of a quantum algorithm that makes $t$ queries in parallel, with no adaptivity. Then, there exists a linear form $T:\mathbb R^{(2n)^t}\to \R$ such that 
    \begin{enumerate}
        \item $a(x)=T((x,1^{n})^{\otimes t})$ 
        for every $x\in\{-1,1\}^n$,
        \item $\norm{T}_{\cb}\leq 1.$
       
    \end{enumerate}
\end{proposition}
\begin{remark}

    For linear forms, we have that $\norm{T}_{\cb}=\norm{T}_\infty.$ However, we stated that result with arbitrary dimension $m\in \N$ to keep the same style as in \cref{theo:cbmethodstd,prop:cbmethodblock,theo:cbmethodboundedadapt}, where $\norm{T}_{\cb}$ is larger than $\norm{T}_\infty$ in general. 
\end{remark}

\section{Optimal root-influence conjecture for completely bounded block-multilinear polynomials} 

In this section, we prove the root-influence type inequality of~\cref{theo:OptimalAAforBMCB}, and two new simulation results of polynomials that exploit the functional inequalities of the kind  of~\cref{theo:OptimalAAforBMCB}. Furthermore, we also show that completely bounded  Bohnenblust--Hille constant for these polynomials is~1. 

\subsection{Proof of \cref{theo:OptimalAAforBMCB} and optimal Bohnenblust-Hille constant}
We start by proving a theorem that quickly implies \cref{theo:optimalcbBHBlockMult} and the optimal Bohnenblust-Hille inequality for block-multilinear completely bounded polynomials. 

\begin{theorem}\label{theo:optimalcbBHBlockMult}

Let $p:\{-1,1\}^n\times\dots\times\{-1,1\}^n\to\R$ be a block-multilinear polynomial of degree $t$. Then, for all $s\in [t]$, it is satisfied that  
\begin{equation}
\label{eq:optimalcbBHBlockMult}
    \norm{p}_{\cb}\geq   \sum_{i\in [n]_0} \left(\sum_{\ind  i\in [n]^t_0:i_s=i} |p_{\ind i}|^2\right)^\frac12. 
\end{equation}  
\end{theorem}

\begin{proof}
For a fixed $s\in [t]$, we write $\ind i=({\ind{i}_L},i_s,{\ind{i}_R})$ with ${\ind{i}_L}=(i_1,\dots, i_{s-1})$ and ${\ind{i}_R}=(i_{s+1},\dots, i_t)$.  We also write 
$$X_1(i_1)\cdots X_{s-1}(i_{s-1})=X_{L}({\ind{i}_L}) \mbox{ and }X_{s+1}(i_{s+1})\cdots X_d(i_{d})=X_R({\ind{i}_R}).$$
It may be that ${\ind{i}_L}$ or ${\ind{i}_R}$ are empty. Then $X_L$ or $X_R$ are just the identity. We will carry out the construction of the unitaries and unit vectors in the space $\mathcal{H}=\bigotimes _{l=1}^t\R^{n+1}.$ We divide the proof in 3 steps.

{\bf Step 1.} First, we define $e$ and $X_l$ for $l\geq s+1.$ 
To do so, consider the  orthonormal basis $$(e_{\ind i}:=e_{i_{1}} \otimes \cdots \otimes e_{i_{t}})_{\ind i\in[n]_0^t},$$ and set $e=e_0\otimes \cdots\otimes e_0$. Then, $X_{k}(i)$ acts only on the factor of $\mathcal{H}$ with $l=k$ (being the identity on the rest of the factors) and satisfies $X_k(i)e_0=e_{i}$.   
From the definition, we see that $(X_{R}({\ind{i}_R})e)_{{{\ind{i}_R}}}$ is an orthonormal set.

 Set $$v_{{\ind{i}_L},i_s}=\sum_{{\ind{i}_R}} p_{({\ind{i}_L},i_s,{\ind{i}_R})} X_R({\ind{i}_R}) e,$$ and observe that, by orthonormality 
\begin{equation}
\label{eq:orthonormalityv}\norm{v_{{\ind{i}_L},s}}_{\ell_2}= \left(\sum_{{\ind{i}_R}}  |p_{({\ind{i}_L},i_s,{\ind{i}_R})}|^2\right)^\frac12.\end{equation}

  {\bf Step 2.} Second, we claim that
\begin{equation}
\label{eq:step1}\norm{p}_{\cb}\geq \sum_{i_s=0}^n \sup _{\substack{\mathcal{F}=(f_{{\ind{i}_L}})_{{\ind{i}_L}}\\
\mathcal{F}\text{ orthonormal}\\|\mathcal{F}|=(n+1)^{s-1}}} \sum_{{\ind{i}_L}} \langle f_{{\ind{i}_L}}, v_{{\ind{i}_L},i_s}\rangle.\end{equation}

To see this, fix $n+1$ orthonormal sets $\mathcal{F}_0, \dots, \mathcal{F}_{n}$ of size $(n+1)^{s-1}$, and set $\mathcal{F}_m=(f^m_{{\ind{i}_L}})_{{\ind{i}_L}}$ for every $m=0,\ldots, n$. A simple computation shows that for any vector $f$ and any matrices  $X_1(i_1),\dots, X_s(i_s)$, we have that 
    \begin{align*}
        \langle f, p(X) e\rangle &=\sum_{i_s=0}^ n \left( \sum_{{\ind{i}_L}} \langle f, X_L({\ind{i}_L}) X_s(i_s) v_{{\ind{i}_L},i_s}\rangle\right) \\
         &=\sum_{i_s=0}^ n \left( \sum_{{\ind{i}_L}} \langle X_s(i_s) ^{\mathsf T} X_{L}({\ind{i}_L})^{\mathsf T} f,   v_{{\ind{i}_L},i_s}\rangle\right).
    \end{align*}

Now, since all the orthonormal sets have the same size, for any $m$ we can find an orthogonal matrix $X_s(m)$ such that $$X_s(m)(f^m_{\ind{i}_L})=f^0_{\ind{i}_L}\hspace{0.4 cm}\text{for every $\ind{i}_L$},$$ with  $X_s(0)=\Id$.

On the other hand, we can easily find orthogonal matrices $X_1(i_1),\ldots, X_{s-1}(i_{s-1})$, with $X_k(0)=\Id$ for $k=1,\ldots, s-1$ such that $$X_L ({i}_L)f^0_{{i}_L}=f_{0},$$ where $f_{0}=f^0_{(0,\dots,0)}$.  Indeed, in order to define these matrices, consider, for every $m=0,\ldots, n$ and $k=1,\ldots, s-1$, the orthonormal sets $$(f^0_{i_1,\ldots , i_{k-1}, m,0,\ldots, 0})_{i_1,\ldots, i_{k-1}=1}^n\hspace{0.3 cm} \text{and}\hspace{0.3 cm}(f^0_{i_1,\ldots , i_{k-1}, 0,0,\ldots, 0})_{i_1,\ldots, i_{k-1}=1}^n.$$ Since they have the same size, we can always find orthogonal matrices, $X_k(m)$ with $X_k(0)=\Id$ for every $k$ and satisfying $$X_k(m)(f^0_{i_1,\ldots , i_{k-1}, m,0,\ldots, 0})=f^0_{i_1,\ldots , i_{k-1}, 0,0,\ldots, 0}.$$ It is clear that these matrices satisfy the desired property.

    Then,

    $$\langle f_{0}, p(X) e\rangle =\sum_{i_s=0}^ n \left( \sum_{{\ind{i}_L}} \langle X_s(i_s) ^{\mathsf T} f^0_{{\ind{i}_L}},   v_{{\ind{i}_L},i_s}\rangle\right)=\sum_{i_s=0}^ n \left( \sum_{{\ind{i}_L}} \langle  f^{i_s}_{{\ind{i}_L}},   v_{{\ind{i}_L},i_s}\rangle\right).$$ 
    As a consequence,
 $$\norm{p}_{\cb}\geq \sum_{i_s=0}^n \sup _{\substack{\mathcal{F}_{i_s}=(f^{i_s}_{{\ind{i}_L}})_{{\ind{i}_L}}\\
  \mathcal{F}_{i_s}\text{ orthonormal}\\|\mathcal{F}_{i_s}|=(n+1)^{s-1}}} \sum_{{\ind{i}_L}} \langle f^{i_s}_{{\ind{i}_L}}, v_{{\ind{i}_L},i_s}\rangle,$$ which is the same as \eqref{eq:step1}.

 {\bf Step 3.} Third, we show that for any vectors $(w_{{\ind{i}_L}})_{{\ind{i}_L}}$

\begin{equation}
    \label{eq:step2}
 \sup _{\substack{\mathcal{F}=(f_{{\ind{i}_L}})_{{\ind{i}_L}}\\
\mathcal{F}\text{ orthonormal}\\|\mathcal{F}|=(n+1)^{s-1}}} \sum_{{\ind{i}_L}} \langle f_{{\ind{i}_L}}, w_{{\ind{i}_L}}\rangle \geq \left(\sum_{{\ind{i}_L}} \|w_{{\ind{i}_L}}\|^2\right)^\frac12.\end{equation}

To see this, for any given $\mathcal{F}$, let $M_{\mathcal{F}}$ be the linear operator on $\mathcal{H}$ defined by $M_{\mathcal{F}}(f_{{\ind{i}_L}})={w_{{\ind{i}_L}}}$ and $M_{\mathcal{F}}(v)=0$ for $v\in \langle \mathcal{F} \rangle ^\perp$. Then,
$$ \sum_{{\ind{i}_L}} \langle f_{{\ind{i}_L}}, w_{{\ind{i}_L}}\rangle =\Tr(M_{\mathcal{F}}).$$
We remark that for any given orthonormal set $\mathcal{F}_0$ we have
$$\{ M_{\mathcal{F}}: \mathcal{F} \mbox{ orthonormal}\}= \{ M_{\mathcal{F}_0} O: O \mbox{ orthogonal}\}.$$ Indeed, let $\mathcal{F}$ be another orthonormal set of $(n+1)^{s-1}$ elements. Then there is an orthogonal matrix $O$ such that $O \mathcal{F}= {\mathcal{F}_0}$, which implies that $M_{\mathcal{F}}=M_{\mathcal{F}_0} O.$ For the reverse inclusion, given $O$ orthogonal, define $\mathcal{F}=O^{-1}\mathcal{F}_0.$

As a consequence, $$ \sup _{\substack{\mathcal{F}=(f_{{\ind{i}_L}})_{{\ind{i}_L}}\\\mathcal{F}\text{ orthonormal}\\|\mathcal{F}|=(n+1)^{s-1}}} \sum_{{\ind{i}_L}} \langle f_{{\ind{i}_L}}, w_{{\ind{i}_L}}\rangle= \sup_{O \mbox{ orthogonal}} \Tr(M_{\mathcal{F}_0}O)=\|M_{\mathcal{F}_0}\|_{S_1}.$$ 
To obtain \eqref{eq:step2} we use the fact that

    \begin{equation}
    \label{eq:stepschatten}
        \|M_{\mathcal{F}_0}\|_{S_1} \geq \|M_{\mathcal{F}_0}\|_{S_2} =  \Tr( M_{\mathcal{F}_0}^{\mathsf T} M_{\mathcal{F}_0})^\frac12= \left(\sum_{{\ind{i}_L}} \|w_{{\ind{i}_L}}\|^2\right)^\frac12.
    \end{equation} Here, $\|M_{\mathcal{F}_0}\|_{S_1}$ stands for the trace norm of the matrix $M_{\mathcal{F}_0}$, which is known to be larger than or equal to its Hilbert-Schmidt norm.
    
Finally, the result follows by combining \eqref{eq:orthonormalityv}, \eqref{eq:step1} and \eqref{eq:step2}.

\end{proof}
\begin{remark}
   From equation \eqref{eq:stepschatten} we see that 
   inequality \eqref{eq:optimalcbBHBlockMult} can be improved to

   $$\|p\|_{\cb} \geq \sum_{i\in [n]_0} \|P^{s,i}\|_{S_1},$$
   for $s\in[t]$, where the operator $P^{s,i}$ is defined by
   $$P^{s,i}(f_{{\ind{i}_L}})=\sum_{{\ind{i}_R}} p_{({\ind{i}_L},i_s,{\ind{i}_R})} e_{{\ind{i}_R}} \mbox{ with } i_s=i,$$ and both $(f_{{\ind{i}_L}})_{{\ind{i}_L}}$ and $(e_{{\ind{i}_R}})_{{\ind{i}_R}}$ are orthonormal sets.

   Therefore, Theorem \ref{theo:optimalcbBHBlockMult} above improves the key technical Lemma 3.3 in \cite{arunachalam2025cb} in two different ways. On the one hand, it is stated for block-multilinear polynomials and not only for multilinear forms; and, on the other hand, it provides a stronger lower bound for the completely bounded norm of the polynomial.
\end{remark}

The first corollary of \cref{theo:optimalcbBHBlockMult} that we prove is \cref{theo:OptimalAAforBMCB}, which we restate for the reader's convenience. 

\theoOptimalAAforBMCB*

\begin{proof}
    By summing \cref{theo:optimalcbBHBlockMult} over all $s\in [t],$ it follows that 
    \begin{equation*}
        t\norm{a}_{\cb}\geq \sum_{s\in[t]}\sum_{i\in [n]_0}\sqrt{\sum_{\ind i\in [n]_0^t:i_s=i}a_{\ind i}^2}\geq \sum_{s\in[t]}\sum_{i\in [n]}\sqrt{\operatorname{Inf}_{x_s(i)}[a]},
    \end{equation*}where in the last inequality we have used Eq. (\ref{inf_block_mult}).

   For the optimality, we remark that $a(x_1,\dots,x_t)=x_1(1)\dots x_t(1)$ satisfies $\norm{a}_{\cb}=1$ and $\sum_{s\in[t]}\sum_{i\in [n]}\sqrt{\Inf_{x_{s}(i)}[a]}=t.$
\end{proof}
The second corollary of
\cref{theo:optimalcbBHBlockMult} is an optimal Bohnenblust-Hille inequality. 
\begin{corollary}
Let $p$ be a block multilinear polynomial of degree $t$. 
\label{coro:bhblockmultilinear}Then,
    \begin{equation}
        \label{eq:BHwithconstant1}
     \left(\sum_{\ind i \in [n]_0^t} |p_{\ind i}|^{\frac{2t}{t+1}}\right)^{\frac{t+1}{2t}} \leq \|p\|_{\cb}.
\end{equation}
The inequality is optimal, as witnessed by $p(x)=x.$
\end{corollary}

\begin{proof}
The result follows from \cref{theo:optimalcbBHBlockMult} and Blei's inequality~\cite[Lemma 5.3]{blei1979fractional}, that reads as 

    \begin{equation*}
        \label{eq:bleiinequality}
  \left(\sum_{\ind i \in [n]_0^t} |p_{\ind i}|^{\frac{2t}{t+1}}\right)^{\frac{t+1}{2t}} \leq \left(\prod_{s\in [t]} \sum_{i\in [n]_0} \left(\sum_{ \substack{\ind i \in [n]_0^t\\i_s=i}} |p_{\ind i}|^2 \right)^{\frac12}\right)^{\frac{1}{t}}.
    \end{equation*} 
\end{proof}

\subsection{Simulation results}\label{sec:simulationresults}

In this section, we state three simulation results for polynomials satisfying certain functional-analytic assumptions, presented from weaker to stronger, and we focus on how stronger assumptions improve the efficiency of the simulation. The first result is due to Aaronson and Ambainis~\cite{Aaronsons:2014}; the second and third are new. We first state these results, and later, in \cref{rem:comparison}, discuss their implications for quantum algorithms such as those in \cref{fig:QQAblocks}. In these results, the restriction of a polynomial $p:\{-1,1\}^n\to \R$ to $x\in\{-1,1\}^S$ for some $S\subseteq[n]$, is the polynomial $p|_{x^S}:~\{-1,1\}^{[n]-S}\to \R$ defined by restricting $p$ to $x^S$. We say that a family of polynomials $\mathcal P$ is closed under restrictions if for every $p\in\mathcal P$ every restriction of $p$ belongs to $\mathcal P.$

\begin{theorem}{\normalfont (\cite[Theorem 22]{Aaronsons:2014})}\label{theo:simAA}
    Let $\mathcal P \subseteq\{p:\{-1,1\}^n\to [-1,1],\ n\in \N\}$  be a family of polynomials closed under restriction. Assume that there exists $w:\N\times [0,\infty)\to \R$ that is non-increasing in the first argument and non-decreasing in the second, and satisfies $$\maxinf[p]\geq w(\deg(p),\var [p])$$ for every $p\in \mathcal P.$ Then, given $p\in\mathcal P$, there is a deterministic classical algorithm that makes, $$\frac{4\deg(p)}{\delta\cdot w(\deg(p),\eps^2\delta/2))}$$ queries and approximates $p$ up to error  $\eps$ on a $(1-\delta)$-fraction of the inputs.
\end{theorem}

The second simulation result applies when an inequality like the one in \cref{eq:InfAA} holds, and its proof is a modification of that of \cref{theo:simAA}, which corresponds to \cref{eq:varAA}.

\begin{theorem}\label{theo:simInf}
    Let $\mathcal P\subset \{p:\{-1,1\}^n\to [-1,1], n \in \mathbb{N}\}$ be a family of polynomials closed under restriction. Assume that there exists a non-decreasing function $w:\N\to \R$ such that $$ \maxinf [p]\geq \frac{\Inf^2[p]}{ w(\deg(p))},$$ for every $p\in \mathcal P.$ Then, given $p\in\mathcal P$, there is a deterministic classical algorithm that makes $$\frac{4 w(\deg(p))}{\delta^2\eps^2}$$ queries and approximates $p$ up to error  $\eps$ on a $(1-\delta)$-fraction of the inputs.    
\end{theorem}

\begin{proof} We give a classical algorithm
$C$ that makes $T=(4/\delta^2\eps^2)\cdot w(\deg(p))$ queries and approximates $p$, up to error $\eps$, on a $(1-\delta)$-fraction of $\{-1,1\}^{n}$. The algorithm on input $x$ behaves as follows. 

\bigskip

\texttt{\qquad set }$p_{0}:=p$

\texttt{\qquad for }$j\in [T-1]_0,$\texttt{:}

\texttt{\qquad\qquad find an }$i\in\left[  {n}-j\right]  $\texttt{\ such
that }$\Inf_i[p_j]=\maxinf[p_j]$

\texttt{\qquad\qquad query }$x_{i}$\texttt{, and let }$p_{j+1}
:\{-1,1\}^{{n}-j}\rightarrow\mathbb{R}$\texttt{\ be the restriction of } $p$\texttt{ to }$x_i$

\qquad \texttt{output }$\widetilde p(x)=\mathbb E[p_{T}]$
\bigskip

Note that $p_T$ is a polynomial that depends on the input $x$, which we assume to be a uniformly random element of $\{-1,1\}^{n}$. It suffices to show that 
\begin{equation}\label{eq:goal}
    \var [p_T]\leq \frac{\delta\eps^2}{2}\quad \text{with probability } 1-\delta/2 \text{ for a uniformly random } x,
\end{equation}
 since, by Chebychev's inequality,  it implies 
\begin{equation*}
    |\widetilde p_T(x)-p(x)|\leq \eps
\end{equation*}
with probability $\geq 1-\delta/2.$ 

Thus, we focus on showing \cref{eq:goal}. We will prove a stronger inequality, namely 
\begin{equation}\label{eq:goal2}
    \Inf [p_T]\leq \frac{\delta\eps^2}{2}\quad \text{with probability } 1-\delta/2 \text{ for a uniformly random } x.
\end{equation}
To this end, we use the notation $p_j(x):\{-1,1\}^{n-j}\to \R$ to refer to the polynomial obtained after restricting $p$ a total of $j$ times when the input is $x$. We also use $i_j(x)$ to refer to the entry that is queried the $j$-th time when the input is $x$. We note that for any $x\in \{-1,1\}^n$ we have

$$\frac12(\Inf [p_j(x)]+\Inf [p_j(x^{\oplus i_{j}(x)})])=\Inf [p_{j-1} (x)]-\Inf_{i_{j}(x)}[p_{j-1}(x)]. $$ 
Now, by the choice of $I_j$ and the hypothesis of the theorem, we have that 
\begin{equation*}
    \frac12(\Inf [p_j(x)]+\Inf [p_j(x^{\oplus i_{j}(x)})])=\Inf[p_{j-1}(x)]-\maxinf[p_j(x)]\leq \Inf[p_{j-1}(x)]-\frac{\Inf^2[p_{j-1}(x)]}{w(\deg(p))}.
\end{equation*}
Taking expectations over the input $x\in \{-1,1\}^n$ on both sides of the above equation, one has that 
\begin{equation*}
    \mathbb E_{x}[\Inf [p_j(x)]]=\E_x\Inf[p_{j-1}(x)]-\frac{\E_x\Inf^2[p_{j-1}(x)]}{w(\deg(p))}.
\end{equation*}
Next, by Jensen's inequality we have that $\E_x[\Inf^2[p_{j-1}(x)]]\geq \E_x[\Inf[p_{j-1}(x)]]^2,$
so 
\begin{equation*}
    \mathbb E_{x}[\Inf [p_j(x)]]\leq \E_x[\Inf[p_{j-1}(x)]]\left(1-\frac{\E_x[\Inf[p_{j-1}(x)]]}{w(\deg(p))}\right).
\end{equation*}
Now, we use \cref{claim:auxi} to ensure that 
\begin{equation}\label{eq:intermediate}
    \mathbb E_{x}[\Inf [p_T(x)]]\leq \frac{\eps^2\delta^2}{4}.
\end{equation}
\begin{claim}\label{claim:auxi}
    Let $K\in (0,\infty)$. Let $(a_j)_{j\in \N}$ be a sequence of nonnegative numbers such that $a_{j+1}\leq a_j(1-a_j/K).$ Then, 
    $$ a_T\leq b,$$
    for $T=\lceil K/b+1\rceil.$
\end{claim}
\begin{claimproof}
    We have that 
    \begin{equation*}
        \frac{1}{a_{j+1}}\geq \frac{1}{a_j(1-a_j/K)}=\frac{1}{a_j}+\frac{1}{K-a_j}\geq \frac{1}{a_j}+\frac{1}{K}.
    \end{equation*}
    Then, 
    \begin{equation*}
        \frac{1}{a_T}\geq \frac{1}{a_1}+\frac{T-1}{K}\geq \frac{T-1}{K}\geq \frac{1}{b},
    \end{equation*}
    as claimed.
 
\end{claimproof}
Finally, we combine \cref{eq:intermediate} with Markov's inequality and arrive at \cref{eq:goal2}.

\end{proof}

The third simulation result is \cref{theo:simfromsqrtinf}. We restate it below for the reader's convenience.

\theosimfromsqrtinf*

\begin{proof}
We define 
\begin{equation*}
    J:=\left\{j\in[n]:\,  \sqrt{\Inf_j[p]}\geq \frac{\eps^2\delta}{\sum_{i\in [n]}\sqrt{\Inf_i[p]}}\right\}. 
\end{equation*}
Then, we define $q:\{-1,1\}^n\to\R$ as 
\begin{equation*}
    q(x_J,x_{[n]-J}):=\mathbb E_{y\in\{-1,1\}^{[n]-J}}[p(x_J,y)].
\end{equation*}
Note that by definition $q$ depends only on the variables of $J$, which satisfy $$ \sum_{i\in [n]}\sqrt{\Inf_i[p]}\geq |J|\cdot \frac{\eps^2\delta}{\sum_{i\in [n]}\sqrt{\Inf_i[p]}},$$
so $|J|\leq (\sum_{i\in [n]}\sqrt{\Inf_i[p]})^2/(\eps^2\delta)$. 

We define \begin{equation}\label{eq:spectrumq}
    q=\sum_{S\subset J}\widehat p(S)\chi_S.
\end{equation}
Now, we can see that $q$ and $p$ are close in $L_2$ distance: 
\begin{align*}
    \mathbb E_x[|q(x)-p(x)|^2]=\sum_{S\not\subseteq J}|\widehat p(S)|^2\leq \sum_{i\not\in J}\Inf_i[p]\leq \frac{\eps^2\delta}{\sum_{j\in [n]}\sqrt{\Inf_j[p]}}\sum_{i\not\in J}\sqrt{\Inf_i[p]}\leq \eps^2\delta,
\end{align*}
where in the first step we have used \cref{eq:spectrumq} and Parseval identity. Finally, since 
\begin{equation*}
    \Pr_x[|q(x)-p(x)|\geq \eps]\cdot \eps^2\leq \E_x[|q(x)-p(x)|^2],
\end{equation*}$q$ defines a classical non-adaptive deterministic algorithm that makes $(\sum_{i\in [n]}\sqrt{\Inf_i[p]})^2/(\eps^2\delta)$ queries and satisfies $$ \Pr_x[|q(x)-p(x)|\geq \eps]\leq \delta.$$
\end{proof}

Now, we are ready to prove \cref{cor:blockQQAsim}. 

\begin{proof}[ of \cref{cor:blockQQAsim}]  According to \cref{theo:OptimalAAforBMCB}, if $a:\{-1,1\}^n\times\dots\times\{-1,1\}^n\to \R$ is the amplitude of a $t$-query quantum algorithm that queries $t$-disjoint blocks of inputs, as in \cref{fig:QQAblocks}, we have 
    \begin{equation*}
     1\geq \norm{a}_{\cb}\geq \sum_{s\in [t]}\sum_{i\in [n]}\frac{\sqrt{\Inf_{x_s(i)}[a]}}{t}.
\end{equation*}

Hence, \cref{theo:simfromsqrtinf} implies the existence of a deterministic classical algorithm that $\eps$-approximates the amplitude on at least a ($1-\delta$)-fraction of the inputs by making at most $t^2/\eps^2\delta$ non-adaptive queries.
\end{proof}

\begin{remark}\label{rem:comparison}
    We now compare the three simulation results for quantum query algorithms as in \cref{fig:QQAblocks}, implied by \cref{theo:simAA,theo:simInf,theo:simfromsqrtinf} and the correspondent state-of-the-art functional inequalities. The Aaronson--Ambainis argument, \cref{theo:simAA}, and the inequality of \cref{eq:varAA} proven in~\cite{gutierrez2023influences} yield a simulation with 
    \begin{equation*}
        \frac{16t^3}{\eps^4\delta^3}
    \end{equation*}
    adaptive classical queries. The inequality of~\cref{eq:InfAA} and \cref{theo:simInf} yield a simulation with 
    \begin{equation*}
        \frac{4t^2}{\eps^2\delta^2}
    \end{equation*}
    adaptive queries, saving a factor $t/\eps^2\delta.$ Finally, \cref{theo:OptimalAAforBMCB,theo:simfromsqrtinf} imply simulation with 
    \begin{equation*}
        \frac{t^2}{\eps^2\delta}
    \end{equation*}
    non-adaptive queries, saving a factor $1/\delta$ and gaining nonadaptivity. 
\end{remark}

\section{Sharp Fourier Growth for high-degree levels of quantum algorithms}
In this section, we show optimal Fourier growth bounds for the $2t$-th level of quantum query algorithms that make $t$ queries. We first treat the standard case, where the queries are made sequentially, as in \cref{fig:QQA}, and then the more general case of bounded adaptivity  (the sequential case is the case with at $t-1$ rounds of adaptivity), as in \cref{fig:QQAboundedadativity}. As the proof of the sequential case is considerably cleaner than the one of the more general case with bounded adaptivity, we have decided to include both in this manuscript, despite that one is subsumed by the other. 
\subsection{The standard case}
In this section we prove \cref{theo:optFGstd}, which we restate for the reader's convenience. Since all polynomials defined on the Boolean cube have degree at most $n$, we note that there are no monomials of degree $2t$ for $2t>n$. Hence, we assume that $2t\leq n$ to prove bounds on the Fourier growth.
\theooptFGstd*
Before proving this theorem, we will prove a very useful lemma that relates the Fourier coefficients of the polynomial $p$ to those of the correspondent multilinear form $T$ defined in \cref{theo:cbmethodstd}.
\begin{lemma}\label{lemma:2}
 Let $T:\mathbb R^{2n}\times\dots\times\R^{2n}\to \R$ be a $2t$-linear form and $p:\{-1,1\}^n\to \R$ defined by $p(x)=T((x,1^n), \dots,(x,1^n))$  for every $x\in\{-1,1\}^n$. Then, for every $S\subseteq [n]$, we have that 
    \begin{equation*}
        \widehat p(S)= \sum_{\ind i:S_{\ind i}=S}T_{\ind i},
    \end{equation*}where for every $\ind i\in [2n]^{2t}$, $T_{\ind i}$ denotes the corresponding coefficient of the $2t$-linear form $T$ and we define  $$S_{\ind i}=\{i\in [n]:\text{ $i$ appears an odd number of times in }\ind i\}.$$
    \end{lemma}
    \begin{proof}
        By linearity, it suffices to prove the particular case where $T$ is a monomial form $$x^1(k_1)\dots x^{2t}(k_{2t}),$$ for some $\ind k\in [2n]^t$. When this form is restricted to $\{-1,1\}^n\times \{1\}^n$ it turns into 
        \begin{equation*}
            p=\chi_{S_{\ind k}},
        \end{equation*}
        so the claim follows.
    \end{proof}

    We are now ready to prove \cref{theo:optFGstd}.

\begin{proof}
    Let $T:\mathbb R^{2n}\times \dots\times \mathbb R^{2n}\to \R$ be the completely bounded $2t$-form of \cref{theo:cbmethodstd} that satisfies $p(x)=T((x,1^n),\dots,(x,1^n))$ for every $x\in\{-1,1\}^n$. 

    We will define a tuple of matrices with operator norm upper bounded by one and evaluate $T$ on them in a way that, combining \cref{theo:cbmethodstd,lemma:2}, we can obtain the desired bound. 
     
     Let $\alpha_S=\sign(\widehat p (S))$ ($\alpha_S=1$ if  $\widehat p (S)=0$) and consider the Hilbert space with orthonormal basis $$\{v\}\cup \{e_S:\, S\subseteq [n], |S|\leq 2t-1\}.$$ 
     
     We define $2n$ matrices acting on that space as
    $$X(i)v=\sum_{i \in S, |S|={2t}} \frac{\alpha_S}{\sqrt{{n-1\choose 2t-1}}} e_{S\setminus i} \mbox{  and  }  X(i)e_S=\delta_{i\in S} e_{S\setminus i},$$ for $1\leq i\leq n$, and $X(i)=\Id$ for $i>n.$ Since it is trivial that $\|X(i)\|=1$ for $i>n$, we only need to study the norm of $X(i)$ for $1\leq i\leq n$. To this end, let $z=a_vv+\sum_{S\subseteq [n]: |S|\leq 2t-1}a_Se_S$ and note that 
    \begin{align*}
X(i)z&=a_v\sum_{S:\, i \in S, |S|={2t}} \frac{\alpha_S}{\sqrt{{n-1\choose 2t-1}}} e_{S\setminus i}+\sum_{S:\, |S|\leq {2t-1}}a_S\delta_{i\in S} e_{S\setminus i}\\&=a_v\sum_{S:\, i \in S, |S|={2t}} \frac{\alpha_S}{\sqrt{{n-1\choose 2t-1}}} e_{S\setminus i}+\sum_{S:\, i \in S, |S|\leq {2t-1}}a_S e_{S\setminus i}.
    \end{align*} Hence,
    \begin{align*}
\|X(i)z\|^2&=\sum_{S:\, i \in S, |S|={2t}} |a_v|^2\frac{1}{{n-1\choose 2t-1}} +\sum_{S:\, i \in S, |S|\leq {2t-1}}|a_S|^2\\&\leq |a_v|^2+\sum_{S:\, |S|\leq {2t-1}}|a_S|^2 =\|z\|^2,
    \end{align*}from where we immediately deduce that $\|X(i)\|\leq 1$.

   In addition, it is easy to check that
    \begin{align}
        \langle e_\emptyset, X(i_1)\dots X(i_{2t}) v\rangle&=\left\{\begin{array}{ll}
           \frac{\alpha_{S_{\ind i}}}{\sqrt{{n-1\choose 2t-1}}}  &  \text{if } |S_\ind i|=2t,\\
           0  & \text{otherwise};
        \end{array}\right.\label{eq:keypropX}
    \end{align}
    note that the condition $|S_\ind i|=2t$ means that the $2t$ elements $i_1,\ldots, i_{2t}\in [n]$ are all different.
    
    Now, by \cref{theo:cbmethodstd} we have that 
    \begin{align*}
        1\geq \sum_{\ind i\in [2n]^{2t}}T_{\ind i}\langle e_{\emptyset},X(i_1)\dots X(i_{2t})v\rangle=\sum_{\ind i:\ |S_{\ind i}|=2t}T_{\ind i} \frac{\alpha_{S_{\ind i}}}{\sqrt{{n-1\choose 2t-1}}}=\frac{1}{\sqrt{{n-1\choose 2t-1}}}\sum_{|S|=2t}\alpha_S\sum_{\ind i:\, S_\ind i=S}T_{\ind i}.
    \end{align*}
    Finally, by \cref{lemma:2} we have that 
    \begin{equation*}
        1\geq \frac{1}{\sqrt{{n-1\choose 2t-1}}}\sum_{|S|=2t}\alpha_S\widehat p(S)=\frac{1}{\sqrt{{n-1\choose 2t-1}}}\norm{\widehat p_{2t}}_{\ell_1}.
    \end{equation*}
    Hence, $$\norm{\widehat p_{2t}}_{\ell_1}\leq \sqrt{{n-1\choose 2t-1}}.$$
\end{proof}
\begin{remark}\label{rem:wherebreaks}
With the proof technique of \cref{theo:optFGstd}, we can also show that the acceptance probability $p$ of $t$-query quantum algorithms satisfies
\begin{equation*}
        \norm{\widehat p_{2t-1}}_{\ell_1}\leq \sqrt{{n-1\choose 2t-2}}.
\end{equation*}
However, we could not adapt our technique to show tight upper bounds to $\norm{\widehat p_{s}}_{\ell_1}$ for $s\in [2t-2]$. The reason why we can only prove upper bounds for the degrees $2t-1$ and $2t$ is that for monomials $\chi_S$ with $|S|=2t-1$ or $|S|=2t$, the set $\{\ind i\in [2n]^{2t}:\, S_\ind i=S\}$ contains only multi-indices $\ind i$ with no repeated elements. Without that, our proof technique breaks. Indeed, the matrices $X(i)$ that we define in the proof of \cref{theo:optFGstd} can be defined for any $s\in [2t]$, but they only satisfy \cref{eq:keypropX} for those $\ind i$ with no repeated elements; if $\ind i$ contains a repeated element, then the LHS of \cref{eq:keypropX} always evaluates to $0$. 
\end{remark}
\subsection{The case of quantum algorithms with bounded adaptivity}

\begin{theorem}\label{prop:fgrowthadaptivequeriesfgrowth}
    Let $p:\{-1,1\}^n\to \mathbb{R}$ be the acceptance probability of an adaptive quantum query algorithm with $r-1$ rounds of adaptivity, performing $t_1,\dots,t_r$ queries in each round, as in \cref{fig:QQAboundedadativity}. Let $t_{\max}=\max t_s$ and $t=t_1+ \dots+t_r$. Then, 
   \begin{equation*}
        \norm{\widehat p_{2t}}_{\ell_1}\leq \left(\frac{2en}{2t- t_{\max}}\right)^{t-t_{\max}/2}.
\end{equation*}
\end{theorem}

As before, we will need the following lemma to relate the Fourier coefficients of the polynomial $p$ to those of the multilinear form $T$ defined in \cref{theo:cbmethodboundedadapt} and whose proof is omitted because it is analogous to that of \cref{lemma:2}.

\begin{lemma}\label{lemma:2adaptive}
Let  $T:\mathbb R^{2nt_1}\times\dots\times \mathbb R^{2nt_N}\to \R$ be an $N$-linear form and let $p(x)=T((x,1^n)^{\otimes t_1},\dots,(x,1^n)^{\otimes t_N})$ for every $x\in\{-1,1\}^n$. Then, for every $S\subseteq [n]$, we have that 
    \begin{equation*}
        \widehat p(S)= \sum_{\ind I:S_{\ind I}=S}T_{\ind I},
    \end{equation*}
    where ${\ind I}=(I_1, \dots, I_{N})\in [2n]^{t_1+\dots+t_{N}} $. 
    \end{lemma}

We are now ready to prove Theorem \ref{prop:fgrowthadaptivequeriesfgrowth}.
\begin{proof}
We define 
    \begin{equation*}
        \tilde t_i=\left\{\begin{array}{ll}
           t_i  & \text{for }1\leq i\leq r, \\
           t_{2r+1-i}  & \text{for } r+1\leq i\leq 2r.
        \end{array}\right.
\end{equation*}
    Let $T:\R^{2n\widetilde t_1}\times \dots\times \R^{2n\widetilde t_{2r}}\to \R$ be the completely bounded $2r$-linear form of \cref{theo:cbmethodboundedadapt}, which satisfies
    \begin{equation*}\label{eq:prestriction}
        p(x)=T((x,1^n)^{\otimes \widetilde{t_1}},\dots,(x,1^n)^{\otimes \widetilde{t_{2r}}}),
    \end{equation*}
    for every $x\in\{-1,1\}^n.$ 
    
    Now, we define a tuple of matrices with bounded operator norm and evaluate $T$ on them in a way that, combining \cref{lemma:2adaptive,theo:cbmethodboundedadapt}, we can obtain the desired bound. We introduce some notation before defining those matrices. Let $\alpha_S=\sign(\widehat p (S))$.
    Let $j_0=\text{argmax}_j t_j$. By the symmetry of $\widetilde t_j$, there are at least two maximizers $j_0.$ Let us also denote $$A:={n-\widetilde t_{j_0}\choose \widetilde t_{j_0+1}+\dots+\widetilde t_{2r}}, \hspace{0.3 cm}B={n-(\widetilde t_{j_0}+\dots+\widetilde t_{2r})\choose \widetilde t_1+\dots+\widetilde t_{j_0-1}}.$$
    Consider the  Hilbert space with orthonormal basis $$\{f_S:S\subseteq [n],\, |S|\leq \widetilde t_{1}+\dots+\widetilde t_{s_0-1}\}\cup \{e_S:S\subseteq [n],\, |S|\leq \widetilde t_{s_0+1}+\dots+\widetilde t_{2r}\}.$$
    We define matrices $X(i)$ acting on that space in the following way. 
   
    For $j_0+1\leq j\leq 2r$ and $I\in [2n]^{\tilde t_j}$ 
 \begin{align*}
 X_j(I)e_S=
\begin{cases}
   e_{S\cup S_{I}} & \text{if } |S_I|=\tilde t_j \text{  and  }S_{I}\cap  S=\emptyset,\\
  0  & \text{otherwise. } 
\end{cases}
    \end{align*}
    For $j=j_0$ and $I\in [2n]^{\tilde t_{j_0}}$ 
   \begin{align*}
X_{j_0}(I)e_S=
\begin{cases}
  \sum_{\substack{S':|S'|= \widetilde t_1+\dots+\widetilde t_{j_0-1},\\S'\cap(S\cup S_{I})=\emptyset}} \frac{\alpha_{S_I\cup S\cup S'}}{\sqrt{AB} }f_{S'} & \text{if } |S_I|=\tilde t_{j_0}, \, |S|= \sum_{j=j_0+1}^{2r}\widetilde t_{j},\, S_I\cap S=\emptyset,\\
  0 & \text{otherwise}.
\end{cases}
\end{align*}
For $1\leq j\leq j_0-1$ and $I\in [2n]^{\tilde t_j}$ 
 \begin{align*}
 X_j(I)f_S=
\begin{cases}
   f_{S\setminus S_{I}} & \text{if } |S_I|=\tilde t_j \text{  and  }S_{I}\subseteq   S,\\
  0  & \text{otherwise. } 
\end{cases}
    \end{align*}

First, we claim that $X_j(I)$ are contractions for every $1\leq j\leq 2r$ and $I\in [2n]^{\tilde t_j}$. This holds for the case $j \neq j_0$, as these matrices map any pair of distinct basis vectors either to zero or to another pair of distinct basis vectors. Let us check the case $j=j_0$. Given  $z=\sum_{S: |S|=\widetilde t_{s_0+1}+\dots+\widetilde t_{2r}}\lambda_S\, e_S$, we have 
\begin{align*}
X_{j_0}(I)z&=\sum_{\substack{S: |S|=\widetilde t_{s_0+1}+\dots+\widetilde t_{2r}\\S\cap S_I=\emptyset}}\lambda_S  \sum_{\substack{S':|S'|= \widetilde t_1+\dots+\widetilde t_{j_0-1},\\S'\cap(S\cup S_{I})=\emptyset}} \frac{\alpha_{S_I\cup S\cup S'}}{\sqrt{AB} }f_{S'} \\&=\sum_{\substack{S':|S'|= \widetilde t_1+\dots+\widetilde t_{j_0-1},\\S'\cap S_{I}=\emptyset}}\Big(\sum_{\substack{S: |S|=\widetilde t_{s_0+1}+\dots+\widetilde t_{2r}\\S\cap (S_I\cup S')=\emptyset}}\lambda_S   \frac{\alpha_{S_I\cup S\cup S'}}{\sqrt{AB} }\Big)f_{S'}.
\end{align*}
Hence, by Cauchy-Schwarz inequality
\begin{align*}\
\|X_{j_0}(I)z\|^2&=\sum_{\substack{S':|S'|= \widetilde t_1+\dots+\widetilde t_{j_0-1},\\S'\cap S_{I}=\emptyset}}\Big|\sum_{\substack{S: |S|=\widetilde t_{s_0+1}+\dots+\widetilde t_{2r}\\S\cap (S_I\cup S')=\emptyset}}\lambda_S   \frac{\alpha_{S_I\cup S\cup S'}}{\sqrt{AB} }\Big|^2\\&\leq \sum_{\substack{S':|S'|= \widetilde t_1+\dots+\widetilde t_{j_0-1},\\S'\cap S_{I}=\emptyset}}\sum_{\substack{S: |S|=\widetilde t_{s_0+1}+\dots+\widetilde t_{2r}\\S\cap (S_I\cup S')=\emptyset}}|\lambda_S|^2\sum_{\substack{S: |S|=\widetilde t_{s_0+1}+\dots+\widetilde t_{2r}\\S\cap (S_I\cup S')=\emptyset}}\frac{1}{AB}\\&\leq \|z\|^2 \sum_{\substack{S':|S'|= \widetilde t_1+\dots+\widetilde t_{j_0-1},\\S'\cap S_{I}=\emptyset}}\sum_{\substack{S: |S|=\widetilde t_{s_0+1}+\dots+\widetilde t_{2r}\\S\cap (S_I\cup S')=\emptyset}}\frac{1}{AB}=\|z\|^2,
\end{align*}where we have used that the number of different disjoint pairs of sets $(S,S')$ of sizes $(\widetilde t_0+\dots\widetilde t_{j_0-1},\widetilde t_{s_j+1}+\dots+\widetilde t_{2r})$ that do not contain any element of $S_I$ is precisely $AB$.

Next, one can verify that for every $\ind I=(I_1,\dots,I_{2r})\in ([2n]^{\widetilde t_1}\times \dots\times [2n]^{\widetilde t_{2r}})$ one has 
\begin{align*}
\langle f_\emptyset, X^1( I_1)\dots X^{2r}(I_{2r})e_\emptyset\rangle=
\begin{cases}
   \frac{\alpha_{S_\ind I}}{\sqrt{AB}} & \text{if } |S_{\ind I}|=2t,\\
  0  & \text{if } |S_{\ind I}|\neq 2t.
\end{cases}
    \end{align*}
    Hence, by \cref{theo:cbmethodboundedadapt}, we have that 
    \begin{equation*}
        1\geq \sum_{\ind I,\, |S_{\ind I}|=2t}T_{\ind I}\frac{\alpha_{S_\ind I}}{\sqrt{AB}}=\sum_{S:\, |S|=2t}\frac{\alpha_{S}}{\sqrt{AB}}\sum_{\ind I,\, S_{\ind I}=S}T_{\ind I}.
    \end{equation*}
    Substituting the values of $A$ and $B$ and applying \cref{lemma:2adaptive}, we obtain

    \begin{equation*}
        \sqrt{{{n-\widetilde t_{j_0}\choose \widetilde t_{j_0+1}+\dots+\widetilde t_{2r}}{n-(\widetilde t_{j_0}+\dots+\widetilde t_{2r})\choose \widetilde t_1+\dots+\widetilde t_{j_0-1}}
        }}\geq \sum_{S,\, |S|=2t} \alpha_S \widehat p(S)=\norm{\widehat p_{2t}}_1.
    \end{equation*}
    Finally, we can use the identity ${n \choose k}\leq (en/k)^k$, the fact that for every $a,b\in \N$ we have that $(1/a)^{a}(1/b)^{b}\leq (2/(a+b))^{a+b}$ and the equality $2t=\sum_{j=1}^{2r}\tilde{t}_{j}$, to write
    \begin{align*}
        &{{n-\widetilde t_{j_0}\choose \widetilde t_{j_0+1}+\dots+\widetilde t_{2r}}{n-(\widetilde t_{j_0}+\dots+\widetilde t_{2r})\choose \widetilde t_1+\dots+\widetilde t_{j_0-1}}
        }\leq {n\choose \widetilde t_{j_0+1}+\dots+\widetilde t_{2r}}{n\choose \widetilde t_1+\dots+\widetilde t_{j_0-1}}
        \\
        &\leq \left(\frac{en}{\widetilde t_{j_0+1}+\dots+\widetilde t_{2r}}\right)^{\widetilde t_{j_0+1}+\dots+\widetilde t_{2r}}\quad \cdot  \left(\frac{en}{\widetilde t_1+\dots+\widetilde t_{j_0-1}}\right)^{\widetilde t_1+\dots+\widetilde t_{j_0-1}}\\
        &\leq \left(\frac{2en}{2t-\widetilde t_{j_0}}\right)^{2t-\widetilde t_{j_0}}=\left(\frac{2en}{2t-t_{\max}}\right)^{2t-t_{\max}}.
    \end{align*}
    Hence, $$\norm{\widehat p_{2t}}_1\leq \left(\frac{2en}{2t-t_{\max}}\right)^{t-t_{\max}/2}.$$

\end{proof}

\section{Classical simulation of non-adaptive quantum query algorithms}
In this section we show that non-adaptive quantum query algorithms, as in \cref{fig:QQAnonadaptive}, can be efficiently simulated by classical non-adaptive query algorithms. The amplitudes of these algorithms are bounded linear forms. Crucially to our proof, we use the fact that for linear forms the $\ell_1$ norm of the tensor coefficients coincides with the completely bounded norm, which does not happen for $k$-linear forms for any $k\geq 2$. 

\begin{fact}\label{fact:boundedlinear}
    Let $T:\R^n\to\R$ be a linear form. Then, 
    \begin{equation*}
        \sup_{x\in\{-1,1\}^n}|\sum_{i\in [n]}T_i x(i)|= \sup_{m\in\N}\sup_{\substack{X(i)\in \R^{m\times m}:\, \norm{X(i)}\leq 1\\i=1,\ldots, n}}\|\sum_{i\in [n]}T_i X(i)\|=\sum_{i\in [n]}|T_{i}|.
    \end{equation*}
\end{fact}
\begin{proof}
    The inequalities
    $\sup|\sum_{i\in [n]}T_i x(i)|\leq \sup|\sum_{i\in [n]}T_i X(i)|\leq \sum_{i\in [n]}|T_{i}|$ are always true, due to inclusion of the feasibility region and the triangle inequality. The inequality $\sum_{i\in [n]}|T_{i}|\leq \sup|\sum_{i\in [n]}T_i x(i)|$ is true because if $y(i)=\text{sgn} (T_i)$, we have that $\sum_{i\in [n]}|T_{i}|= \sum_{i\in [n]}T_iy(i).$
\end{proof}

Together with \cref{prop:cbmethodnonadaptive}, \cref{fact:boundedlinear} implies that the $\ell_1$-norm of the Fourier coefficients of the amplitudes of non-adaptive quantum query algorithms are bounded. This corresponds to \cref{prop:ell1boundamplitudes}, which we restate for the reader's convenience. 

\propell*

\begin{proof}
    By \cref{prop:cbmethodnonadaptive} and \cref{fact:boundedlinear}, there exists a linear form $T:\R^{(2n)^t}\to \R$ such that
    \begin{equation*}
        T((x,1^n),\ldots, (x,1^n))=a(x)
    \end{equation*}
    for every $x\in\{-1,1\}^n$ and $\sum_{\ind i\in [2n]^t}|T_{\ind i}|=\|T\|_{\cb}\leq 1$.
    
    Using (cf. Lemma \ref{lemma:2adaptive}) that for every $S\subseteq [n],$ we have that 
    \begin{equation*}
        \widehat a(S)=\sum_{\ind i:\, S_{\ind i}=S}T_{\ind i},
    \end{equation*} we deduce that $$\sum_S|\widehat a(S)|=\sum_S|\sum_{\ind i:\, S_{\ind i}=S}T_{\ind i}|\leq \sum_{\ind i\in [2n]^t}|T_{\ind i}|\leq 1.$$
\end{proof}

Now, we are ready to prove the simulation result for non-adaptive quantum query algorithms that measure all qubits and only accepts on $r$ of the measurements outcomes. This result is \cref{cor:simnonadapt}, which we restate for the reader's convenience. 

\corsimnonadapt*

\begin{proof}
     Such a $p$ can be written as sum of $r$ polynomials $a_i^2$, with $a_i$ as in \cref{prop:ell1boundamplitudes}, and consequently, satisfying $\|\widehat{a_i}\|_{\ell_1}\leq 1$. This implies that
     \begin{align*}
         \|\widehat{p}\|_{\ell_1}&\leq \sum_{i=1}^r \|\widehat{(a_i)^2}\|_{\ell_1} = \sum_{i=1}^r \|\widehat{a_i}*\widehat{a_i}\|_{\ell_1}\leq \sum_{i=1}^r \|\widehat{a_i}\|^2_{\ell_1} \leq r.
     \end{align*}

We now make use of the following simulation result.
\begin{claim}
\label{lemma:simulation}
   Let $p:\{-1,1\}^n\to \R$ be a polynomial of degree $t$. Then, for every $\varepsilon,\delta>0,$ there exists a non-adaptive randomized classical algorithm that makes only $$O(t\log(1/\delta)/\varepsilon^2)$$ queries, and that, for every input $x \in \{-1,1\}^n$, outputs an estimate $\widetilde p (x)$ that, with probability $\geq 1-\delta$, satisfies $|\widetilde{p}(x)-p(x)|\leq \varepsilon \|\widehat{p}\|_1$.
\end{claim}
\begin{claimproof}
     Assume $\norm{\widehat p}_{\ell_1}\leq 1$, and the general case follows by rescaling. We define the following probability $q$ on the finite set $\{\perp\}\cup\mathcal P([n])$:
    $$ q(S):=\left\{\begin{array}{ll}
        |\widehat p(S)| & \text{if }S\subseteq [n], \\
        1-\norm{\widehat p}_{\ell_1} & \text{if }S=\perp.
    \end{array}\right.$$ 
    The classical algorithm starts by sampling $T=O(\log(1/\delta)/\varepsilon^2)$ elements $S_1,\dots,S_T$ according to $q,$ and we define 
    \begin{equation*}
        \widetilde p(x):=\frac{1}{T}\sum_{r\in [T]:S_r\subseteq [n]}\sign(\widehat p(S_r))\chi_{S_r}(x).
    \end{equation*} 
    The algorithm finishes by evaluating $\widetilde p(x)$. To do that, as $\deg(p)\leq t$, we just need to query $t\cdot T$ entries of $x.$ Furthermore, we can make all these queries at the same time, without making adaptive choices. 
    
    Finally, we claim that $|\widetilde p(x)-p(x)|\leq \eps$ with probability $\geq 1-\delta$. Indeed, we define the random variable 
    \begin{equation*}
        P_{x}(S):=\left\{\begin{array}{ll}
        \sign(\widehat p(S))\chi_{S}(x) & \text{if }S\subseteq [n], \\
        0 & \text{if }S=\perp.
    \end{array}\right.
    \end{equation*}
    Note that we have that $\mathbb E_{q}[P_{x}]=p(x)$ and that $P_{x}$ takes values in~$[-1,1]$. 
    Then, by a Hoeffding bound, we have that $|\widetilde p(x)-p(x)|\leq \eps$ with probability $1-\delta$, as desired. 
    \end{claimproof}
    To finish the proof of \cref{cor:simnonadapt}, we set $\varepsilon=\varepsilon/r$ in Claim \ref{lemma:simulation}.
\end{proof}

\textbf{Acknowledgements.} We acknowledge the use of ChatGPT Pro in the proof of \cref{theo:simfromsqrtinf}, which was proved in conversations where we asked how to use \cref{eq:goal2} to prove a junta theorem similar to the Friedgut junta theorem~\cite{friedgut1998boolean}. Claude was used for language editing. 
\bibliographystyle{alpha}
\bibliography{Bibliography}
\end{document}